\documentclass[letterpaper,11pt]{article}

\usepackage[margin=1in]{geometry}  
\usepackage{enumerate}
\usepackage[fixpdftex]{xcolor}
\usepackage[title]{appendix}

\usepackage{amsmath} 
\usepackage{amssymb}
\usepackage{amsfonts}
\usepackage{amsthm}
\usepackage{mathrsfs}
\usepackage{color, xcolor}
\usepackage{prettyref}
\usepackage{indentfirst}
\usepackage{graphicx}
\usepackage{xspace}
\usepackage{subfig}

\usepackage{ifthen}
\newboolean{fullversion}
\setboolean{fullversion}{TRUE} 
\newboolean{slides}
\setboolean{slides}{FALSE}		
\newboolean{conf}
\setboolean{conf}{FALSE} 
\usepackage{stmaryrd}

\newtheorem*{RSP}{Replica Symmetry Postulate}

\usepackage{empheq}

\definecolor{myblue}{rgb}{.8, .8, 1}
\definecolor{mathblue}{rgb}{0.2472, 0.24, 0.6} 
\definecolor{mathred}{rgb}{0.6, 0.24, 0.442893}
\definecolor{mathyellow}{rgb}{0.6, 0.547014, 0.24}

\usepackage{soul}
\definecolor{lightblue}{rgb}{.90,.95,1}
\sethlcolor{lightblue}

\usepackage[
            CJKbookmarks=true,
            bookmarksnumbered=true,
            bookmarksopen=true,
            colorlinks=true,
            citecolor=red,
            linkcolor=blue,
            anchorcolor=red,
            urlcolor=blue,
            pdfauthor={Yihong Wu}
            ]{hyperref}

\usepackage{pgf}
\usepackage{authblk}
\usepackage[percent]{overpic}
\usepackage{rotating}

\usepackage{tikz}
\usetikzlibrary{arrows}
\tikzstyle{int}=[draw, fill=blue!20, minimum size=2em]
\tikzstyle{dot}=[circle, draw, fill=blue!20, minimum size=2em]
\tikzstyle{init} = [pin edge={to-,thin,black}]

\newcommand{\eg}{e.g.\xspace}
\newcommand{\ie}{i.e.\xspace}
\newcommand{\iid}{i.i.d.\xspace}

\newcommand{\mmse}{{\mathsf{mmse}}}

\newcommand{\XGn}{X_{\sf G}}

\newcommand{\reals}{{\mathbb{R}}}

\newcommand{\naturals}{{\mathbb{N}}}

\newcommand{\eexp}{{\rm e}}

\newcommand{\diff}{{\rm d}}

\newcommand{\Bigo}[1]{{O\left(#1\right)}}
\newcommand{\smallo}[1]{{o\left(#1\right)}}

\newcommand{\Span}{{\rm span}}
\newcommand{\rank}{\mathop{\sf rank}}
\newcommand{\Tr}{\mathop{\sf Tr}}

\newcommand{\Expect}{\mathbb{E}}
\newcommand{\expect}[1]{\mathbb{E}\left[ #1 \right]}

\newcommand{\Prob}{\mathbb{P}}
\newcommand{\prob}[1]{{ \mathbb{P}\left\{ #1 \right\} }}

\newcommand{\toprob}{\xrightarrow{\Prob}}
\newcommand{\tolp}[1]{\xrightarrow{L^{#1}}}
\newcommand{\toas}{\xrightarrow{{\rm a.s.}}}

\newcommand{\var}{\mathsf{var}}

\newcommand{\argmin}{\mathop{\rm argmin}}

\newcommand{\trans}{^{\rm T}}
\newcommand{\Th}{{^{\rm th}}}

\ifthenelse{\boolean{slides}}{
\theoremstyle{remark}
\newtheorem{remark}{Remark}}{
\theoremstyle{plain}

\newtheorem{lemma}{Lemma}
\newtheorem{theorem}{Theorem}

\theoremstyle{definition}
\newtheorem{definition}{Definition}

\newtheorem{remark}{Remark}

}

\theoremstyle{plain}

\newrefformat{eq}{(\ref{#1})}
\newrefformat{chap}{Chapter~\ref{#1}}
\newrefformat{sec}{Section~\ref{#1}}
\newrefformat{algo}{Algorithm~\ref{#1}}
\newrefformat{fig}{Fig.~\ref{#1}}
\newrefformat{tab}{Table~\ref{#1}}
\newrefformat{rmk}{Remark~\ref{#1}}
\newrefformat{clm}{Claim~\ref{#1}}
\newrefformat{def}{Definition~\ref{#1}}
\newrefformat{cor}{Corollary~\ref{#1}}
\newrefformat{lmm}{Lemma~\ref{#1}}
\newrefformat{prop}{Proposition~\ref{#1}}
\newrefformat{app}{Appendix~\ref{#1}}
\newrefformat{ex}{Example~\ref{#1}}
\newrefformat{exer}{Exercise~\ref{#1}}
\newrefformat{soln}{Solution~\ref{#1}}

\newcommand{\lunder}[1]{{\underset{\raise0.3em\hbox{$\smash{\scriptscriptstyle-}$}}{#1}}}

\newcommand{\floor}[1]{{\left\lfloor {#1} \right \rfloor}}
\newcommand{\ceil}[1]{{\left\lceil {#1} \right \rceil}}

\newcommand{\norm}[1]{\left\|{#1} \right\|}

\newcommand{\lnorm}[2]{\left\|{#1} \right\|_{{#2}}}

\newcommand{\Fnorm}[1]{\lnorm{#1}{\rm F}}

\newcommand{\indc}[1]{{\mathbf{1}_{\left\{{#1}\right\}}}}

\def\innergetnumber#1[#2]#3{#2}
\def\getnumber{\expandafter\innergetnumber\jobname}

\newcommand{\Bdim}{{\dim_{\rm B}}}
\newcommand{\uBdim}{\overline{\dim}_{\rm B}}

\newcommand{\od}{{\overline{d}}}
\newcommand{\ud}{{\underline{d}}}

\newcommand{\bsA}{{\boldsymbol{A}}}
\newcommand{\bsB}{{\boldsymbol{B}}}

\newcommand{\bsH}{{\boldsymbol{H}}}

\newcommand{\sft}{{\mathsf{t}}}

\newcommand{\sfA}{{\mathsf{A}}}

\newcommand{\sfD}{{\mathsf{D}}}

\newcommand{\sfR}{{\mathsf{R}}}

\newcommand{\sfT}{{\mathsf{T}}}

\newcommand{\calA}{{\mathcal{A}}}
\newcommand{\calB}{{\mathcal{B}}}

\newcommand{\calD}{{\mathcal{D}}}

\newcommand{\calN}{{\mathcal{N}}}

\newcommand{\calR}{{\mathcal{R}}}

\newcommand{\scrD}{{\mathscr{D}}}
\newcommand{\oscrD}{{\overline{\scrD}}}
\newcommand{\uscrD}{{\underline{\scrD}}}

\newcommand{\bfA}{{\mathbf{A}}}

\newcommand{\bfH}{{\mathbf{H}}}
\newcommand{\bfI}{{\mathbf{I}}}

\newcommand{\bfK}{{\mathbf{K}}}

\newcommand{\bfU}{{\mathbf{U}}}

\newcommand{\hx}{{\hat{x}}}

\newcommand{\hX}{{\hat{X}}}
\newcommand{\hY}{{\hat{Y}}}

\newcommand{\spt}[1]{{\rm spt}(#1)}
\newcommand{\comp}[1]{{#1^{\rm c}}}

\newcommand{\ntok}[2]{{#1,\ldots,#2}}

\newcommand{\restrict}[2]{\left.#1\right|_{#2}}

\newcommand{\Lip}{\mathrm{Lip}}

\newcommand{\renyi}{R\'enyi\xspace}

\newcommand{\rlip}{{\sfR}}
\newcommand{\rll}{{\hat{\sfR}}}

\newcommand{\rlinear}{{\sfR^*}}

\newcommand{\Dstar}{D^*}
\newcommand{\DLstar}{D^*_{\rm L}}
\newcommand{\DL}{D_{\rm L}}
\newcommand{\Rstar}{\calR^*}
\newcommand{\RLstar}{\calR^*_{\rm L}}
\newcommand{\RL}{\calR_{\rm L}}
\newcommand{\zstar}{\zeta^*}
\newcommand{\zLstar}{\zeta^*_{\rm L}}
\newcommand{\zL}{\zeta_{\rm L}}
\newcommand{\xstar}{\xi^*}
\newcommand{\xLstar}{\xi^*_{\rm L}}
\newcommand{\xL}{\xi_{\rm L}}

\xdefinecolor{darkgreen}{rgb}{0,0.35,0}

\newcommand{\pth}[1]{\left( #1 \right)}

\newcommand{\sth}[1]{\left\{ #1 \right\}}

\newcommand{\snr}{{\mathsf{snr}}}
\newcommand{\qsnr}{{\sqrt{\mathsf{snr}}}}

\author{Yihong Wu\thanks{\href{mailto:yihongwu@princeton.edu}{yihongwu@princeton.edu}} }
\author{Sergio Verd\'u\thanks{\href{mailto:verdu@princeton.edu}{verdu@princeton.edu}}}
\affil{Department of Electrical Engineering
\\Princeton University
\\Princeton NJ, 08540}

\title{Optimal Phase Transitions in Compressed Sensing
}
\date{\today}

\begin{document}

\maketitle

\let\oldthefootnote\thefootnote
\renewcommand{\thefootnote}{\fnsymbol{footnote}}
\footnotetext[1]{The results of this paper were presented in part at 
the Third Annual School of Information Theory, University of Southern California, Los Angeles CA, August 5 -- 8, 2010 \cite{noisyCS.poster}
and the IEEE International Symposium on Information Theory, Cambridge, MA, July 1--6, 2012 \cite{optimal.pt.cs.isit}.}
\let\thefootnote\oldthefootnote

\begin{abstract}
Compressed sensing deals with efficient recovery of analog signals from linear encodings. This paper presents a statistical study of compressed sensing by modeling the input signal as an \iid process with known distribution. Three classes of encoders are considered, namely optimal nonlinear, optimal linear and random linear encoders. Focusing on optimal decoders, we investigate the fundamental tradeoff between measurement rate and reconstruction fidelity gauged by error probability and noise sensitivity in the absence and presence of measurement noise, respectively. The optimal phase transition threshold is determined as a functional of the input distribution and compared to suboptimal thresholds achieved by popular reconstruction algorithms. In particular, we show that Gaussian sensing matrices incur no penalty on the phase transition threshold with respect to optimal nonlinear encoding. Our results also provide a rigorous justification of previous results based on replica heuristics in the weak-noise regime.

\vspace{1em}
\end{abstract}

\section{Introduction}
\label{sec:intro}

\subsection{Setup}
	\label{sec:setup}
Compressed sensing \cite{CSfirst,donoho} is a signal processing technique that compresses analog vectors by means of a linear transformation.
By leveraging prior knowledge of the signal structure (\eg, sparsity) and by designing efficient nonlinear reconstruction algorithms, effective compression is achieved by taking a much smaller number of measurements than the dimension of the original signal.

An abstract setup of compressed sensing is shown in \prettyref{fig:cs.setup}: A real vector $x^n \in \reals^n$ is mapped into $y^k \in \reals^k$ by an encoder (or compressor) $f: \reals^n \to \mathbb{R}^k$. The decoder (or decompressor) $g: \reals^k \to \reals^n$ receives $\hat{y}^k$, a possibly noisy version of the measurement, and outputs $\hx^n$ as the reconstruction. The \emph{measurement rate}, \ie, the dimensionality compression ratio, is given by
	\begin{equation}
	R = \frac{k}{n}. 
	\label{eq:rate}
\end{equation}

\begin{figure}[ht]
	\centering
\begin{tikzpicture}[scale=1.2,transform shape,node distance=2.5cm,auto,>=latex']
    \node [int] (f) {$\substack{\textrm{~~~~Encoder~~~~} \\ f:~\mathbb{R}^n \to \mathbb{R}^k}$};
    \node (start) [left of=f,node distance=2cm, coordinate] {};
    \node [dot, pin={[init]above:$e^k$}] (noise) [right of=f] {$+$};    
    \node [int] (g) [right of=noise] {$\substack{\textrm{~~~~Decoder~~~~} \\ g:~\mathbb{R}^k \to \mathbb{R}^n}$};    
    \node [coordinate] (end) [right of=g, node distance=2cm]{};
    \path[->] (start) edge node {$x^n$} (f);
    \path[->] (f) edge node {$y^k$} (noise);
    \path[->] (noise) edge node {$\hat y^k$} (g);
    \draw[->] (g) edge node {$\hat x^n$} (end) ;
\end{tikzpicture}
	\caption{Compressed sensing: an abstract setup.}
	\label{fig:cs.setup}
\end{figure}
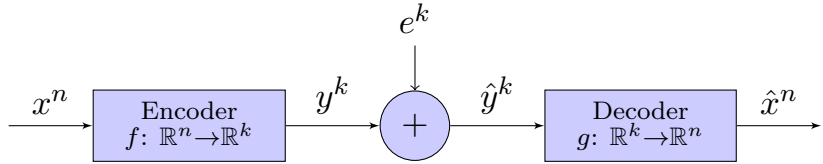

Most of the compressed sensing literature focuses on the setup where
\begin{enumerate}[a)]
	\item performance is measured on a worst-case basis with respect to $x^n$.
	\item the encoder is constrained to be a linear mapping characterized by a $k \times n$ matrix $\bfA$, called the \emph{sensing} or \emph{measurement matrix}, which is usually assumed to be random, and known at the decoder.\footnote{Alternative notations have been used to denote the signal dimension and the number of measurements, \eg, 
	$(m,n)$ in \cite{donoho} and $(N,K)$ in \cite{CT06}.}
	\item the decoder is a low-complexity algorithm which is robust with respect to observation noise, for example, 	decoders based on convex optimizations such as $\ell_1$-minimization \cite{CDS99} and $\ell_1$-penalized least-squares (\ie\ LASSO) \cite{lasso}, greedy algorithms such as matching pursuit \cite{matching.pursuit}, graph-based iterative decoders such as approximate message passing (AMP) \cite{maleki.pnas}, fast iterative shrinkage-thresholding algorithm (FISTA) \cite{fista}, etc.
\end{enumerate}

In contrast, in this paper we formulate an information-theoretic fundamental limit in the following setup:
\begin{enumerate}[a)]
	\item the input vector $x^n$ is random with a known distribution and performance is measured on an average basis.\footnote{Similar Bayesian modeling is followed in some of the compressed sensing literature, for example, \cite{donoho.tanner.2009, maleki.pnas, renyi.ITtrans, maleki.noise.sens, dongning.noisyCS, reeves.gastpar, RG12, KWT10,KMSSZ11, TCSV11}.}
	\item in addition to the performance that can be achieved by the optimal sensing matrix, we also investigate the optimal performance that can be achieved by any \emph{nonlinear} encoder.  
	\item the decoder is \emph{optimal}:\footnote{The performance of optimal decoders for support recovery in the noisy case has been studied in \cite{Wainwright09, fletcher.supp, AT10} on a worst-case basis.}
	\begin{itemize}
	\item In the noiseless case, it is required to be \emph{Lipschitz continuous} for the sake of robustness; 
	\item In the noisy case, it is the minimum mean-square error (MMSE) estimator, \ie, the conditional expectation of the input vector given the noisy measurements.
\end{itemize}	
\end{enumerate}
Due to the constraints of actual measuring devices in certain applications of compressed sensing (\eg, MRI \cite{CS.MRI}, high-resolution radar imaging \cite{CS.radar}), one does not have the freedom to optimize over all possible sensing matrices. Therefore we consider both optimized as well as random sensing matrices and investigate their respective fundamental limits achieved by the corresponding optimal decoders.

  \subsection{Phase transition}
	\label{sec:pt}
	The general goal is to investigate the fundamental tradeoff between reconstruction fidelity and measurement rate as $n\to\infty$, as a functional of the signal and noise statistics.
	
	When the measurements are noiseless, the goal is to reconstruct the original signal as perfectly as possible by driving the error probability to zero as the ambient dimension, $n$, grows. For many input processes, \eg, independent and identically distributed (\iid) ones, it turns out that there exists a threshold for the measurement rate, above which it is possible to achieve a vanishing error probability and below which the error probability will eventually approach one for any sequence of encoder-decoder pairs. Such a phenomenon is known as \emph{phase transition} in statistical physics. In information-theoretic parlance, we say that the \emph{strong converse} holds.
	
	When the measurement is noisy, exact analog signal recovery is obviously impossible and we gauge the reconstruction fidelity by
the \emph{noise sensitivity}, defined as the ratio between the mean-square reconstruction error and the noise variance. Similar to the behavior of error probability in the noiseless case, there exists a phase transition threshold of measurement rate, which only depends on the input statistics, above which the noise sensitivity is bounded for all noise variances, and below which the noise sensitivity blows up as the noise variance tends to zero.

\subsection{Signal model}
\label{sec:model}
Sparse vectors, supported on a subspace with dimension smaller than $n$, play an important role in signal processing and statistical models. 
A stochastic model that captures sparsity is the following mixture distribution
\cite{renyi.ITtrans,donoho.tanner.hypercube, maleki.pnas, dongning.noisyCS, KWT10,KMSSZ11, TCSV11}:
\begin{equation}
P = (1-\gamma) \delta_{0} + \gamma \, P_{\rm c},
	\label{eq:sparseP}
\end{equation}
where $\delta_0$ denotes the Dirac measure at 0, $P_{\rm c}$ is a probability measure absolutely continuous with respect to the Lebesgue measure, and $0 \leq \gamma \leq 1$. Consider a random vector $X^n$ independently drawn from $P$. By the weak law of large numbers, 
$\frac{1}{n} \lnorm{X^n}{0} \toprob \gamma$, where the ``$\ell_0$ norm'' $\lnorm{\cdot}{0}$ denotes the number of non-zeros of a vector. This corresponds to the regime of \emph{proportional (or linear) sparsity}. In \prettyref{eq:sparseP}, the weight on the continuous part $\gamma$ parametrizes the signal sparsity
and $P_{\rm c}$ serves as the prior distribution of non-zero entries.


Generalizing \prettyref{eq:sparseP}, we henceforth consider \emph{discrete-continuous mixed} distributions (\ie, \emph{elementary distributions} \cite{renyibook}):
\begin{equation}
	P_X = (1-\gamma) P_{\rm d} + \gamma P_{\rm c},
	\label{eq:dcmix}
	\end{equation}
where $P_{\rm d}$ is a discrete probability measure and $P_{\rm c}$ is an absolutely continuous probability measure. 
For simplicity we focus on \iid input processes in this paper. 
Note that apart from sparsity, there are other signal structures that have been previously explored in the compressed sensing literature. For example, the so-called \emph{simple} signal in infrared absorption spectroscopy \cite[Example 3, p. 914]{donoho.tanner.ieee10} is such that each entry of the signal vector is constrained to lie in the unit interval, with most of the entries saturated at the boundaries (0 or 1). Similar to the rationale that leads to \prettyref{eq:sparseP}, an appropriate statistical model for simple signals is a mixture of a Bernoulli distribution and an absolutely continuous distribution supported on the unit interval, which is a particular instance of \prettyref{eq:dcmix}. Although most of the results in the present paper hold for arbitrary input distributions, with no practical loss of generality, we will be focusing on discrete-continuous mixtures (\ie, without singular components) because of their relevance to compressed sensing applications.

\subsection{Main contributions}
	\label{sec:results}
	We introduced the framework of almost lossless analog compression in \cite{renyi.ITtrans} as a Shannon-theoretic formulation of noiseless compressed sensing. Under regularity conditions on the encoder or the decoder, \cite{renyi.ITtrans} derives various coding theorems for the minimal measurement rate involving the information dimension of the input distribution, introduced by Alfr\'ed \renyi in 1959 \cite{renyi}. Along with the Minkowski and MMSE dimension, we summarize a few relevant properties of \renyi information dimension in \prettyref{sec:dim}. 	The most interesting regularity constraints are the linearity of the compressor and Lipschitz continuity (robustness) of the decompressor, which are considered \emph{separately} in \cite{renyi.ITtrans}. \prettyref{sec:noiseless} gives a brief summary of the non-asymptotic version of these results. In addition, in this paper we also consider the fundamental limit when linearity and Lipschitz continuity are imposed \emph{simultaneously}. For \iid discrete-continuous  mixtures, we show that the minimal measurement rate is given by the input information dimension, \ie, the weight $\gamma$ of the absolutely continuous part. Moreover, the Lipschitz constant of the decoder can be chosen independently of $n$, as a function of the gap between the measurement rate and $\gamma$. This results in the optimal phase transition threshold of error probability in noiseless compressed sensing.
	
	Our main results are presented in \prettyref{sec:noisy}, which deals with the case where the measurements are corrupted by additive Gaussian noise. We consider three formulations of noise sensitivity: optimal nonlinear, optimal linear and random linear (with \iid entries) encoder and the associated optimal decoder. In the case of \iid input processes, we show that for any input distribution, the phase transition threshold for optimal encoding is given by the input information dimension. Moreover, this result also holds for discrete-continuous mixtures with optimal linear encoders and Gaussian random measurement matrices. Invoking the results in \cite{mmse.dim.IT}, we show that the calculation of the reconstruction error with random measurement matrices based on \emph{heuristic} replica methods in \cite{dongning.noisyCS} predicts the correct phase transition threshold. These results also serve as a rigorous verification of the replica calculations in \cite{dongning.noisyCS} in the high-SNR regime (up to $o(\sigma^2)$ as the noise variance $\sigma^2$ vanishes).
	
	The fact that randomly chosen sensing matrices turn out to incur no penalty in phase transition threshold with respect to optimal nonlinear encoders lends further importance to the conventional compressed sensing setup described in \prettyref{sec:setup}.
	
	In \prettyref{sec:compare}, we compare the optimal phase transition threshold to the suboptimal threshold of several practical reconstruction algorithms under various input distributions. In particular, we demonstrate that the thresholds achieved by
 the $\ell_1$-minimization decoder and the AMP decoder \cite{donoho.tanner.hypercube,maleki.noise.sens} lie far from the optimal boundary, especially in the highly sparse regime which is most relevant to compressed sensing applications.
	

\section{Three dimensions}
\label{sec:dim}
In this section we introduce three dimension concepts for sets and probability measures involved in various coding theorems in Sections \ref{sec:noiseless} and \ref{sec:noisy}.

\subsection{Information dimension}
\label{sec:renyi}


A key concept in fractal geometry, in \cite{renyi} R\'enyi defined the \emph{information dimension} (also known as the \emph{entropy dimension} \cite{peres.conformal}) of a probability distribution. It measures the rate of growth of the entropy of successively finer discretizations.
\begin{definition}
Let $X$ be a real-valued random variable. Let $m \in \naturals$. 
The \emph{information dimension} of $X$ is defined as
\begin{equation}
	d(X) = \lim_{m \to \infty} \frac{H\left(\floor{m X}\right)}{\log m}.
	\label{eq:renyid}
\end{equation}
If the limit in \prettyref{eq:renyid} does not exist, the $\liminf$ and $\limsup$ are called lower and upper information dimensions of $X$ respectively, denoted by $\underline{d}(X)$ and $\overline{d}(X)$.
	\label{def:renyidim}
\end{definition}

\prettyref{def:renyidim} can be readily extended to random vectors, where the floor function $\floor{\cdot}$ is taken componentwise. Since $d(X)$ only depends on the distribution of $X$, we also denote $d(P_X) = d(X)$. The same convention also applies to other information measures.


The information dimension of $X$ is finite if and only if the mild condition
\begin{equation}
H(\floor{X}) < \infty	
	\label{eq:dfinite}
\end{equation}
is satisfied \cite{renyi.ITtrans}. A sufficient condition for $d(X) < \infty$ is $\expect{\log(1 + |X|)} < \infty$, much milder than finite mean or finite variance. 

Equivalent definitions of information dimension include:\footnote{The lower and upper information dimension are given by the $\liminf$ and $\limsup$ respectively.}
\begin{itemize}
	\item 
For an integer $M \geq 2$, write the $M$-ary expansion of $X$ as
\begin{equation}
X = \floor{X} + \sum_{i \geq 1} (X)_i M^{-i},
	\label{eq:Mary.expansion}
\end{equation}
where the $i\Th$ $M$-ary digit $(X)_i \triangleq \floor{M^i X} - M \floor{M^{i-1} X}$ is a discrete random variable taking values in $\{\ntok{0}{M-1}\}$. Then $d(X)$ is the normalized entropy rate of the digits $\{(X)_i\}$:
\begin{equation}
	d(X) = \lim_{m \to \infty} \frac{H(\ntok{(X)_1}{(X)_m})}{m \log M}.
	\label{eq:drate}
\end{equation}

\item
Denote by 
$B(x,\delta)$ the open ball of radius $\delta$ centered at $x$. Then (see \cite[Definition 4.2]{hk.proj} and \cite[Appendix A]{renyi.ITtrans})
\begin{gather}
d(X) = \lim_{\delta \downarrow 0} \frac{ \expect{\log P_X(B(X,\delta))}}{\log \delta}. \label{eq:hunt}
\end{gather}

\item The rate-distortion function of $X$ with mean-square error distortion is given by
\begin{equation}
R_X(D) = \inf \limits_{\Expect |\hX - X|^2 \leq D} I(X;\hX).
\label{eq:RD}
\end{equation}
Then \cite[Proposition 3.3]{dembo}
\begin{gather}
d(X) = \lim_{D \downarrow 0} \frac{R_X(D)}{\frac{1}{2}\log \frac{1}{D}}. 
\label{eq:dembo}
\end{gather}
\item 
Let $N \sim \calN(0,1)$ be independent of $X$. The mutual information $I(X; \qsnr X + N)$ is finite if and only if \prettyref{eq:dfinite} holds \cite{mmse.functional.IT}. 
Then \cite{guionnet}
\begin{gather}
 d(X) = \lim_{\snr \to \infty} \frac{I(X; \qsnr X + N)}{\frac{1}{2} \log \snr}. 
 \label{eq:i.renyi}
\end{gather}
	
\end{itemize}

The alternative definition in \prettyref{eq:drate} implies that $d(X^n) \leq n$ (as long as it is finite).
For discrete-continuous mixtures, the information dimension is given by the weight of the absolutely continuous part.
\begin{theorem}[{\cite{renyi}}]
	Assume that $X$ has a discrete-continuous mixed distribution as in \prettyref{eq:dcmix}. If $H(\floor{X}) < \infty$, then
\begin{equation}
d(X) = \gamma.
\end{equation}
	\label{thm:dmix}
\end{theorem}
In the presence of a singular component, the information dimension does not admit a simple formula in general. One example where the information dimension can be explicitly determined is the \emph{Cantor distribution}, which can be defined via the following ternary expansion
\begin{equation}
	X = \sum_{i \geq 1} (X)_i 3^{-i},
	\label{eq:cantor}
\end{equation}
where $(X)_i$'s are \iid and equiprobable on $\{0,2\}$. Then $P_X$ is absolutely singular with respect to the Lebesgue measure and $d(X) = \log_3 2 \approx 0.63$, in view of \prettyref{eq:drate}.

\subsection{MMSE dimension}
	\label{sec:mmsed}
	
	Introduced in \cite{mmse.dim.IT}, the MMSE dimension is an information measure that governs the high-SNR asymptotics of the MMSE in Gaussian noise. Denote the MMSE of estimating $X$ based on $Y$ by
\begin{align}
	\mmse(X|Y) 
& ~ = \inf_f \expect{(X - f(Y))^2} \label{eq:m0}\\ 
& ~ =  \expect{(X - \Expect[X|Y])^2} = \expect{\var(X|Y)},
	\label{eq:m}
\end{align}
where the infimum in \prettyref{eq:m0} is over all Borel measurable $f$.
When $Y$ is related to $X$ through an additive Gaussian noise channel with gain $\qsnr$, \ie, $Y = \qsnr X + N$ with $N \sim \calN(0,1)$ independent of $X$, we denote
\begin{equation}
	\mmse(X, \snr) = \mmse(X | \qsnr X + N).
	\label{eq:mmse}
\end{equation}
\begin{definition}
The \emph{MMSE dimension} of $X$ is defined as
\begin{equation}
	\scrD(X) = \lim_{\snr \to \infty} \snr \cdot \mmse(X,\snr).
	\label{eq:mmsedim}
\end{equation}
Useful if the limit in \prettyref{eq:mmsedim} does not exist, the $\liminf$ and $\limsup$ are called lower and upper MMSE dimensions of $X$ respectively, denoted by $\uscrD(X)$ and $\oscrD(X)$.
	\label{def:mmsedim}
\end{definition}

It is shown in \cite[Theorem 8]{mmse.dim.IT} that the information dimensions are sandwiched between the MMSE dimensions: if \prettyref{eq:dfinite} is satisfied, then
\begin{equation}
	0 \leq \uscrD(X) \leq \ud(X) \leq \od(X) \leq \oscrD(X) \leq 1. 	
	\label{eq:mmse.renyi}
\end{equation}
For discrete-continuous mixtures, the MMSE dimension coincides with the information dimension:
\begin{theorem}[{\cite[Theorem 15]{mmse.dim.IT}}]
	If $X$ has a discrete-continuous mixed distribution as in \prettyref{eq:dcmix}, then $\scrD(X) = \gamma$.
	\label{thm:Dmix}
\end{theorem}
It is possible that the MMSE dimension does not exist and the inequalities in \prettyref{eq:mmse.renyi} are strict. For example, consider the Cantor distribution in \prettyref{eq:cantor}. Then the product $\snr \cdot \mmse(X,\snr)$ oscillates periodically in $\log \snr$ between $\uscrD(X) \approx 0.62$ and $\oscrD(X) \approx 0.64$ \cite[Theorem 16]{mmse.dim.IT}.

\subsection{Minkowski dimension}
In fractal geometry, the Minkowski dimension (also known as the box-counting dimension) \cite{falconer2} gauges the fractality of a subset in metric spaces, defined as the exponent with which the covering number grows. The ($\epsilon$-)Minkowski dimension of a probability measure is defined as the lowest Minkowski dimension among all sets with measure at least $1-\epsilon$ \cite{pesin}.
\begin{definition}[Minkowski dimension]
Let $A$ be a nonempty bounded subset of $\reals^n$. For $\delta > 0$, denote by $N_{A}(\delta)$ the $\delta$-covering number of $A$, \ie, the smallest number of $\ell_2$-balls of radius $\delta$ needed to cover $A$. Define the (upper) Minkowski dimension of $A$ as
\begin{align}
	\uBdim{A} = \limsup_{\delta \to 0} \frac{\log N_{A}(\delta)}{\log \frac{1}{\delta}}.
	\label{eq:ubdim}
\end{align}
Let $\mu$ be a probability measure on $(\reals^n, \calB_{\reals^n})$. Define the ($\epsilon$-)Minkowski dimension of $\mu$ as
\begin{gather}
\uBdim^{\epsilon}(\mu) = \inf\{ \uBdim(A): \mu(A) \geq 1 - \epsilon\}.
\label{eq:ubdim2}
\end{gather}
\label{def:bdim}
\end{definition}

Minkowski dimension is always nonnegative and less than the ambient dimension $n$, with $\Bdim A  = 0$ for any finite set $A$ and $\Bdim A  = n$ for any bounded set $A$ with nonempty interior. An intermediate example is the middle-third Cantor set $C$ in the unit interval: $\Bdim C = \log_3 2$ \cite[Example 3.3]{falconer2}.

\section{Noiseless compressed sensing}
\label{sec:noiseless}

\subsection{Definitions}

\begin{definition}[Lipschitz continuity]
Let $U \subset \reals^n$ and $f: U \to \reals^k$. Define\footnote{Throughout the paper, $\norm{\cdot}$ denotes the $\ell_2$ norm on the Euclidean space. It should be noted that the proof in the present paper relies crucially on the inner product structure endowed by the $\ell_2$ norm. See \prettyref{rmk:hilbert}.} 
\begin{equation}
\Lip(f) \triangleq  \sup_{x \neq y} \frac{\norm{f(x) - f(y)}}{\norm{x -y}}.
\end{equation}
If $\Lip(f) \leq L$ for some $L \in \reals_+$, we say that $f$ is \emph{$L$-Lipschitz continuous}, and $\Lip(f)$ is called the \emph{Lipschitz constant} of $f$.
	\label{def:lipschitz}
\end{definition}

\begin{remark}
	$\Lip(\cdot)$ defines a pseudo-norm on the space of all functions.
	\label{rmk:pseudo}
\end{remark}

The Shannon-theoretic fundamental limits of noiseless compressed sensing are defined as follows.
\begin{definition}
Let $X^n$ be a random vector consisting of independent copies of $X$.
Define the minimum $\epsilon$-achievable rate to be the minimum of $R > 0$ such that there exists a sequence of encoders $f_n: \reals^n \to \reals^{\floor{Rn}}$ and decoders $g_n: \reals^{\floor{Rn}} \to \reals^n$, such that
\begin{equation}
\prob{g_n(f_n(X^n)) \neq X^n} \leq \epsilon
\label{eq:err.prob}
\end{equation}
for all sufficiently large $n$. The minimum $\epsilon$-achievable rate is denoted by $\rlinear(X,\epsilon), \rlip(X,\epsilon)$  and $\rll(X,\epsilon)$ depending on the class of allowable encoders and decoders as specified in \prettyref{tab:def}.\footnote{It was shown in \cite{renyi.ITtrans} that in the definition of $\rlinear$ and $\rlip$, the continuity constraint can be replaced by Borel measurability without changing the minimum rate.}
\begin{table}[!ht]
\renewcommand{\arraystretch}{1.3}
\caption{Regularity conditions of encoder/decoders and corresponding minimum $\epsilon$-achievable rates.}
\centering
\begin{tabular}{|l|l|c|}
	\hline
	Encoder & Decoder & Minimum $\epsilon$-achievable rate\\
	\hline
	\hline
	Linear & Continuous & $\rlinear(X, \epsilon)$\\
	\hline
	Continuous & Lipschitz & $\rlip(X, \epsilon)$\\
	\hline
	Linear & Lipschitz & $\rll(X, \epsilon)$\\
	\hline	
	\end{tabular}
	\label{tab:def}
\end{table}
	\label{def:main}
\end{definition}

\begin{remark}
In \prettyref{def:main}, $\rlip$ and $\rll$ are defined under the Lipschitz continuity assumption of the decoder, which does not preclude the case where the Lipschitz constants blow up as the dimension grows. For practical applications, decoders with bounded Lipschitz constants are desirable, which amounts to constructing a sequence of decoders with Lipschitz constant only depending on the rate and the input statistics. As we will show later, this is indeed possible for discrete-continuous mixtures.
\end{remark}

\subsection{Results}
	\label{sec:results.noiseless}
	
	The following general result holds for any input process \cite{renyi.ITtrans}:
\begin{theorem}
For any $X$ and any $0 \leq \epsilon \leq 1$, 
\begin{equation}
	\rlinear(X,\epsilon) \leq \rlip(X,\epsilon) \leq \rll(X,\epsilon) .
	\label{eq:ordering}
\end{equation}	
Moreover, \prettyref{eq:ordering} holds for arbitrary input processes that are not necessarily \iid.
	\label{thm:ordering}
\end{theorem}
The second inequality in \prettyref{eq:ordering} follows from the definitions, since 
\[
\rll(X, \epsilon) \geq \max\{\rlinear(X, \epsilon), \rlip(X, \epsilon)\}.
\]
Far less intuitive is the first inequality, proved in \cite[Section V]{renyi.ITtrans}, which states that robust reconstruction is always harder to achieve than linear compression.

 The following result is a finite-dimensional version of the general achievability result of linear encoding in \cite[Theorem 18]{renyi.ITtrans}, which states that sets of low Minkowski dimension can be linearly embedded into low-dimensional Euclidean space probabilistically. This is a probabilistic generalization of the embeddability result in \cite{hk.linear}.
\begin{theorem}
	Let $X^n$ be a random vector with $\uBdim^{\epsilon}(P_{X^n}) \leq k$.
	Let $m > k$. Then for Lebesgue almost every $\bfA \in \reals^{m \times n}$, there exists a $\pth{1-\frac{k}{m}}$-H\"older continuous function $g: \reals^m \to \reals^n$, \ie, $\|g(x)-g(y)\| \leq L \|x-y\|^{1-\frac{k}{m}}$ for some $L >0$ and all $x,y$, such that $\prob{g(\bfA X^n) \neq X^n} \leq \epsilon$.
	\label{thm:dimB.finite}
\end{theorem}

\begin{remark}
	In \prettyref{thm:dimB.finite}, the decoder can be chosen as follows: by definition of $\uBdim^{\epsilon}$, there exists $U \subset \reals^n$, such that $\uBdim(U) \leq k$. Then if $x^n$ is the \emph{unique} solution to the linear equation $\bfA x^n = y^k$ in $U$, the decoder outputs $g(y^k) = x^n$; otherwise $g(y^k) = 0$.
	\label{rmk:dec}
\end{remark}

Generalizing \cite[Theorem 9]{renyi.ITtrans}, a non-asymptotic converse for Lipschitz decoding is the following:
\begin{theorem}
	For any random vector $X^n$, if there exists a Borel measurable $f: \reals^n \to \reals^k$ and a Lipschitz continuous $g: \reals^k \to \reals^n$ such that $\prob{g(f(X^n)) \neq X^n} \leq \epsilon$, then
\begin{equation}
k \geq \uBdim^{\epsilon}(P_{X^n}) \geq \od(X^n)  - \epsilon n.
\label{eq:lipconv.f}
\end{equation}
	\label{thm:lipconv.f}
\end{theorem}
\begin{proof}
	\prettyref{sec:pf.a}.
\end{proof}
\begin{remark}
	An immediate consequence of \prettyref{eq:lipconv.f} is that for general input processes, we have
\begin{equation}
\rlip(X, \epsilon) \geq \limsup_{n \to \infty} \frac{\od(X^n)}{n}  - \epsilon.
\label{eq:lipconv.weak}
\end{equation}
which, for \iid inputs, becomes
\begin{equation}
\rlip(X, \epsilon) \geq \od(X)  - \epsilon.
\label{eq:lipconv.weak.m}
\end{equation}
In fact, combining the left inequality in \prettyref{eq:lipconv.f} and the following concentration-of-measure result \cite[Theorem 14]{renyi.ITtrans}: for any $0 < \epsilon < 1$,
\begin{equation}
	\liminf_{n\to\infty} \frac{\uBdim^{\epsilon}(P_{X^n})}{n} \geq \od(X),
	\label{eq:mink.conv}
\end{equation}
\prettyref{eq:lipconv.weak.m} can be superseded by the following \emph{strong converse}: 
\begin{equation}
\rlip(X, \epsilon) \geq \od(X) .
\label{eq:lipconv}
\end{equation}
General achievability results for $\rlip(X, \epsilon)$ rely on rectifiability results from geometric measure theory \cite{federer}. See \cite[Section VII]{renyi.ITtrans}.
	\label{rmk:lip}
\end{remark}

For discrete-continuous mixtures, we show that linear encoders and Lipschitz decoders can be realized \emph{simultaneously} with \emph{bounded} Lipschitz constants.
\begin{theorem}[Linear encoding: discrete-continuous mixture]
Let $P_X$ be a discrete-continuous mixed distribution of the form \prettyref{eq:dcmix}, with the weight of the continuous part equal to $\gamma$.	Then
\begin{equation}
	\rlinear(X, \epsilon) = {d}(X) = \gamma
	\label{eq:rlinear.dc}
	\end{equation}
for all $0 < \epsilon <1$. Moreover, if the discrete part $P_{\rm d}$ has finite entropy, then for any rate $R > d(X)$, the decompressor can be chosen to be Lipschitz continuous with respect to the $\ell_2$-norm with a Lipschitz-constant independent of $n$:
\begin{equation}
	L = \frac{\sqrt{R}}{R-\gamma} \pth{\frac{R}{\gamma}}^{\frac{\gamma}{R-\gamma}} \exp\pth{ \frac{H(P_{\rm d}) (1 - \gamma)+ h(\gamma)}{R-\gamma} + \frac{1}{2}},
	\label{eq:lipconst}
\end{equation}
Consequently, 
\begin{equation}
	\rll(X, \epsilon) = \rlinear(X, \epsilon) = {d}(X) = \gamma.
	\label{eq:rll.dc}
\end{equation}
	\label{thm:linear.mix}
\end{theorem}
\begin{proof}
 \prettyref{sec:pf.a}.
\end{proof}

Combining \prettyref{thm:linear.mix} and \cite[Theorem 10]{renyi.ITtrans} yields the following tight result: for any \iid input with a common distribution of the discrete-continuous mixture form in \prettyref{eq:dcmix}, whose discrete component has finite entropy, we have
\begin{equation}
	\rlinear(X, \epsilon) = \rlip(X, \epsilon) = \rll(X, \epsilon) = \gamma
	\label{eq:tight}
	\end{equation}
for all $0 < \epsilon <1$. In the special case of sparse signals ($P_{\rm d} = \delta_0$) with $s = \gamma n$ non-zeros, this implies that roughly $s$ linear measurements are sufficient to recover the unknown vector with high probability. This agrees with the well-known result that $s+1$ measurements are both necessary and sufficient to reconstruct an $s$-sparse vector probabilistically (see, \eg, \cite{FB96}).

\begin{remark}
In the achievability proof of \prettyref{thm:linear.mix}, our construction of a sequence of Lipschitz decoders with bounded Lipschitz constants independent of the dimension $n$ only works for recovery performance measured in the $\ell_2$ norm. 
	The reason is two-fold: First, the Lipschitz constant of a linear mapping with respect to the $\ell_2$ norm is given by its maximal singular value, whose behavior for random measurement matrices is well studied. Second, Kirszbraun's theorem states that any Lipschitz mapping from a subset of a Hilbert space to a Hilbert space  can be extended to the whole space with the same Lipschitz constant \cite[Theorem 1.31, p. 21]{NLFA.Schwartz}. This result fails for general Banach spaces, in particular, for $\reals^n$ equipped with any $\ell_p$-norm ($p \neq 2$) \cite[p. 20]{NLFA.Schwartz}. Of course, by the equivalence of norms on finite-dimensional spaces, it is always possible to extend to a Lipschitz function with a larger Lipschitz constant; however, such soft analysis does not control the size of the Lipschitz constant, which may possibly blow up as the dimension increases.
	Nevertheless, \prettyref{eq:lipconv} shows that even if we allow a sequence of decompressors with Lipschitz constants that diverges as $n \to \infty$, the compression rate is still lower bounded by $\od(X)$.
\label{rmk:hilbert}
\end{remark}

\begin{remark}[Behavior of the Lipschitz constant]
The Lipschitz constant of the decoder is a proxy to gauge the decoding robustness. It it interesting to investigate what is the smallest attainable Lipschitz constant as a for a given rate $R > \gamma$. Note that the constant in \prettyref{eq:lipconst} depends exponentially on $\frac{1}{R-\gamma}$, which implies that the decoding becomes increasingly less robust as the rate approaches the fundamental limit.
For sparse signals ($P_{\rm d} = \delta_0$ hence $H(P_{\rm d})=0$), 
\prettyref{eq:lipconst} reduces to
\begin{align}
L 
= & ~ \frac{\sqrt{\eexp R}}{R-\gamma} \pth{ \frac{R}{\gamma^2}}^{\frac{\gamma}{R-\gamma}} \pth{ \frac{1}{1-\gamma}}^{\frac{1-\gamma}{R-\gamma}}.	\label{eq:lipconst.sparse}
\end{align}
It is unclear whether it is possible to achieve a Lipschitz constant that diverges polynomially as $R \to \gamma$.

	\label{rmk:lipconst}
\end{remark}

\begin{remark}

Although too computationally intensive and numerically unstable (in fact discontinuous in general), in the conventional compressed sensing setup, the optimal decoder is an $\ell_0$-minimizer that seeks that sparsest solution compatible with the linear measurements. In our Bayesian setting, such a decoder does not necessarily minimize the probability of selecting the wrong signal.
However, the $\ell_0$-minimization decoder does achieve the asymptotic fundamental limit $\sfR^*(X,\epsilon)$ for any sparse $P_X = (1-\gamma) \delta_0 + \gamma P_{\rm c}$, since it is, in fact, even better than the asymptotically optimum decoder described in Remark 3. 
The optimality of the $\ell_0$-minimization decoder for sparse signals has also been observed in \cite[Section IV-A1]{KWT10} based on replica heuristics.
	\label{rmk:l0}
\end{remark}

\begin{remark}
Converse results for any linear encoder and decoder pair have been proposed before in other compressed sensing setups. For example, the result in \cite[Theorem 3.1]{DIPW10} assumes noiseless measurement with arbitrary sensing matrices and recovery algorithms, dealing with 
best sparse approximation under $\ell_1/\ell_1$-stability guarantee. The following non-asymptotic lower bound on the number of measurements is shown: if there exist a sensing matrix $\bfA \in \reals^{k \times n}$, a decoder $g: \reals^k \to B_0^n(s)$\footnote{$B_0^n(s) = \{x\in \reals^n: \|x\|_0 \leq s\}$ denotes the collection of all $s$-sparse $n$-dimensional vectors.} and a constant $C>0$, such that
\begin{equation}
\|x-g(\bfA x)\|_1 \leq \min_{z \in B_0^n(s)} C \|x-z\|_1
\end{equation}
for any $x \in \reals^n$, then $k \geq \frac{s \log \frac{n}{2s}}{\log(4+2C)}$.
However, this result does not directly apply to our setup because we are dealing with $\ell_2/\ell_2$-stability guarantee with respect to the measurement noise, instead of the sparse approximation error of the input vector.
	\label{rmk:doba}
\end{remark}

\section{Noisy compressed sensing}
\label{sec:noisy}

%
%

%

%


\subsection{Setup}
\label{sec:noisy.cs}

The basic setup of noisy compressed sensing is a joint source-channel coding problem as shown in \prettyref{fig:noisy.cs}, 
\begin{figure}[ht]
	\centering
\begin{tikzpicture}[scale=1.2,transform shape,node distance=2.5cm,auto,>=latex']
    \node [int] (a) {$\substack{\textrm{~~~~Encoder~~~~} \\ f_n:~\mathbb{R}^n \to \mathbb{R}^k}$};
    \node (b) [left of=a,node distance=2cm, coordinate] {};
    \node [dot, pin={[init]above:{$\sigma N^k$}}] (d) [right of=a] {$+$};    
    \node [int] (c) [right of=d] {$\substack{\textrm{~~~~Decoder~~~~} \\ g_n:~\mathbb{R}^k \to \mathbb{R}^n}$};    
    \node [coordinate] (end) [right of=c, node distance=2cm]{};
    \path[->] (b) edge node {$X^n$} (a);
    \path[->] (a) edge node {$Y^k$} (d);
    \path[->] (d) edge node {$\hat Y^k$} (c);
    \draw[->] (c) edge node {$\hat X^n$} (end) ;
\end{tikzpicture}
	\caption{Noisy compressed sensing setup.}
	\label{fig:noisy.cs}
\end{figure}
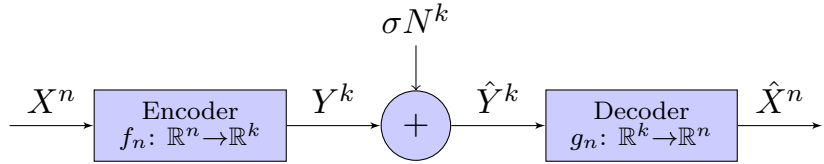
 where we assume that

\begin{itemize}
	\item The source $X^n$ consists of \iid copies of a real-valued random variable $X$ with unit variance.
	\item The channel is stationary memoryless with \iid additive Gaussian noise $\sigma N^k$ where $N^k \sim \calN(0, \bfI_k)$.
	\item 
	Unit average power constraint on the encoded signal:
	\begin{equation}
	\frac{1}{k} \Expect[\lnorm{f_n(X^n)}{2}^2] \leq 1.
	\label{eq:avg.power}
\end{equation}

	
	\item The reconstruction error is gauged by the per-symbol MSE distortion:
	\begin{equation}
	d(x^n, \hat{x}^n) = \frac{1}{n} \|\hat{x}^n-x^n\|_{2}^2.
	\label{eq:distortion.mse}
\end{equation}
	
\end{itemize}

In this setup, the fundamental question is: For a given noise variance and measurement rate, what is the lowest reconstruction error? For a given encoder $f$, the corresponding optimal decoder $g$ is the \emph{MMSE estimator} of the input $X^n$ given the channel output $\hat Y^k = f(X^n) + \sigma N^k$. Therefore the optimal distortion achieved by encoder $f$ is
\begin{equation}
 \inf_g \expect{\|X^n - g(\hY^k)\|^2} =	\mmse(X^n | f(X^n) + \sigma N^k).
	\label{eq:distortion.f}
\end{equation}

In the case of noiseless compressed sensing, the interesting regime of measurement rates is between zero and one.
 When the measurements are noisy, in principle it makes sense to consider measurement rates greater than one in order to combat the noise. Nevertheless, the optimal phase transition for noise sensitivity is always less than one, because with $k=n$ and an invertible measurement matrix, the linear MMSE estimator achieves bounded noise sensitivity for any noise variance.



\subsection{Distortion-rate tradeoff}
For a fixed noise variance $\sigma^2$, we define three distortion-rate functions that correspond to \emph{optimal} encoding, \emph{optimal linear} encoding and \emph{random linear} encoding, respectively. In the remainder of this section, we fix $k = \floor{R n}$.

%
%
%
%

\subsubsection{Optimal encoder}
\begin{definition}
The minimal distortion achieved by the optimal encoding scheme is given by:
\begin{align}
	\Dstar(X, R,\sigma^2) \triangleq & ~ \limsup_{n \to \infty} \frac{1}{n} \inf_f
	\sth{\mmse(X^n | f(X^n)+ \sigma N^k)\colon \Expect[\lnorm{f(X^n)}{2}^2] \leq k}.
	\label{eq:Dstar}
\end{align}
	\label{def:Dstar}
\end{definition}
For stationary ergodic sources, the asymptotic optimization problem in \prettyref{eq:Dstar} can be solved by applying Shannon's joint source-channel coding separation theorem \cite[Section XI]{shannon49}, which states that the lowest rate, $R$, that achieves distortion $D$ is given by
\begin{equation}
	R = \frac{R_X(D)}{C(\sigma^2)},
	\label{eq:sepa}
\end{equation}
where $R_X(\cdot)$ is the rate-distortion function of $X$ in \prettyref{eq:RD} and $C(\sigma^2) = \frac{1}{2} \log (1+\sigma^{-2})$ is the AWGN channel capacity. By the monotonicity of the rate-distortion function, we have
\begin{equation}
	\Dstar(X, R,\sigma^2) = R_X^{-1}\left(\frac{R}{2} \log(1+\sigma^{-2})\right).
	\label{eq:Dstar.sepa}
\end{equation}
In general, optimal joint source-channel encoders are nonlinear \cite{massey.JSC}. In fact, Shannon's separation theorem states that the composition of an optimal lossy source encoder and an optimal channel encoder is asymptotically optimal when blocklength $n \to \infty$. Such a construction results in an encoder that is finite-valued, hence nonlinear. For fixed $n$ and $k$, linear encoders are in general suboptimal. 

\subsubsection{Optimal linear encoder}
To analyze the fundamental limit of conventional noisy compressed sensing, we restrict the encoder $f$ to be a linear mapping, denoted by a matrix $\bfH \in \reals^{k \times n}$. Since $X^n$ are \iid with zero mean and unit variance, the input power constraint \prettyref{eq:avg.power} simplifies to
	\begin{equation}
\Expect[\lnorm{\bfH X^n}{2}^2] = \Expect[{X^n}\trans \bfH\trans \bfH X^n] =	\Tr(\bfH\trans \bfH) = \Fnorm{\bfH}^2 \leq k,
	\label{eq:power.tr}
\end{equation}
where $\Fnorm{\cdot}$ denotes the Frobenius norm. 
\begin{definition}
Define the optimal distortion achievable by linear encoders as:
\begin{align}
\DLstar(X, R,\sigma^2) 
\triangleq & ~ \limsup_{n \to \infty} \frac{1}{n} \inf_{\bfH} \sth{ \mmse(X^n | \bfH X^n+ \sigma N^k)\colon \Fnorm{\bfH}^2 \leq k}.
	\label{eq:DLstar}
\end{align}
	\label{def:DLstar}
\end{definition}

\subsubsection{Random linear encoder}
We consider the ensemble performance of random linear encoders and relax the power constraint in \prettyref{eq:power.tr} to hold on average:
\begin{equation}
\Expect[\Fnorm{\bsA}^2] \leq k.
	\label{eq:power.travg}
\end{equation}
In particular, we focus on the following ensemble of random sensing matrices, for which \prettyref{eq:power.travg} holds with equality:

\begin{definition}
Let $\bsA_n$ be a $k \times n$ random matrix with \iid entries of zero mean and variance $\frac{1}{n}$. The minimal expected distortion achieved by this ensemble of linear encoders is given by:
\begin{align}
\DL(X, R,\sigma^2)
\triangleq & ~  \limsup_{n \to \infty} \frac{1}{n} \mmse(X^n| (\bsA_n X^n + \sigma N^k, \bsA_n)) \label{eq:DL} \\
= & ~ \limsup_{n \to \infty} \mmse(X_1| (\bsA_n X^n + \sigma N^k, \bsA_n)) \label{eq:DL.1}
\end{align}
where \prettyref{eq:DL.1} follows from symmetry and $\mmse(\cdot | \cdot)$ is defined in \prettyref{eq:m}.\footnote{The MMSE on the right-hand side of \prettyref{eq:DL} and \prettyref{eq:DL.1} can be computed by first fixing the sensing matrix $\bsA_n$ then averaging with respect to its distribution.}
\label{def:DL}
\end{definition}

General formulae for $\DL(X, R,\sigma^2)$ and $\DLstar(X, R,\sigma^2)$ are yet unknown. One example where they can be explicitly computed is given in \prettyref{sec:gaussian.jscc} -- the Gaussian source.

\subsubsection{Properties}
\begin{theorem}
\begin{enumerate}
\item For fixed $\sigma^2$, $\Dstar(X, R,\sigma^2)$ and $\DLstar(X, R,\sigma^2)$ are both decreasing, convex and continuous in $R$ on $(0,\infty)$. 
\label{prop1}

\item For fixed $R$, $\Dstar(X, R,\sigma^2)$ and $\DLstar(X, R,\sigma^2)$ are both decreasing, convex and continuous in $\frac{1}{\sigma^2}$ on $(0,\infty)$.
\label{prop2}


	\item
	\begin{equation}
	\Dstar(X, R,\sigma^2) \leq \DLstar(X, R,\sigma^2) \leq \DL(X, R,\sigma^2) \leq 1.
	\label{eq:rank}
\end{equation}
\label{prop3}
\end{enumerate}

\label{thm:distortion.property}
\end{theorem}


\begin{proof}
\begin{enumerate}
\item 
Fix $\sigma^2$. Monotonicity with respect to the measurement rate $R$ is straightforward from the definition of $\Dstar$ and $\DLstar$. Convexity follows from time-sharing between two encoding schemes. Finally, convexity on the real line implies continuity.

\item Fix $R$. For any $n$ and any encoder $f\colon \reals^n \to \reals^k$, $\sigma^2 \mapsto \mmse(X^n | f(X^n) + \sigma N^k)$ is increasing. This is a consequence of the infinite divisibility of the Gaussian distribution as well as the data processing lemma of MMSE \cite{zamir.fisher}. Consequently, $\sigma^2 \mapsto \Dstar(X, R,\sigma^2)$ is also increasing. Since $\Dstar$ can be equivalently defined as
\begin{equation}
	\Dstar(X, R,\sigma^2) = \limsup_{n \to \infty} \frac{1}{n} \inf_f \sth{\mmse(X^n | f(X^n)+ N^k): \Expect[\lnorm{f(X^n)}{2}^2] \leq \frac{k}{\sigma^2} },
	\label{eq:Dstar.alt}
\end{equation}
convexity in $\frac{1}{\sigma^2}$ follows from time-sharing. The results on $\DLstar$ follows analogously.

\item
	The leftmost inequality in \prettyref{eq:rank} follows directly from the definition, while the rightmost inequality follows because we can always discard the measurements and use the mean as an estimate.
	Although the best sensing matrix will beat the average behavior of any ensemble, the middle inequality in \prettyref{eq:rank} is not quite trivial because the power constraint in \prettyref{eq:avg.power} is not imposed on each matrix in the ensemble. The proof of this inequality can be found in \prettyref{app:middle}.
 \qedhere
	 \end{enumerate}
\end{proof}

\begin{remark}
	Alternatively, the convexity properties of $\Dstar$ can be derived from \prettyref{eq:Dstar.sepa}. Since $R_X(\cdot)$ is decreasing and concave, $R_X^{-1}(\cdot)$ is decreasing and convex, which, composed with the concave mapping $\sigma^{-2} \mapsto \frac{R}{2} \log (1+\sigma^{-2})$, gives a convex function $\sigma^{-2} \mapsto \Dstar(X, R,\sigma^2)$ \cite[p. 84]{boyd}. The convexity of $R \mapsto \Dstar(X, R,\sigma^2)$ can be similarly proved.
\end{remark}

\begin{remark}
Note that the time-sharing proofs of \prettyref{thm:distortion.property}.\ref{prop1} and \ref{thm:distortion.property}.\ref{prop2} do not work for $\DL$, because time-sharing between two random linear encoders results in a block-diagonal matrix with diagonal submatrices each filled with \iid entries. This ensemble is outside the scope of random matrices with \iid entries considered in \prettyref{def:DL}. Therefore, proving the convexity of $R \mapsto \DL(X, R,\sigma^2)$ amounts to showing that replacing all zeroes in the block-diagonal matrix with independent entries of the same distribution always helps with the estimation. This is certainly not true for individual matrices.
\end{remark}

\subsection{Phase transition of noise sensitivity}
One of the main objectives of noisy compressed sensing is to achieve robust reconstruction, obtaining a reconstruction error that is proportional to the noise variance. To quantify robustness, we analyze \emph{noise sensitivity}, namely the ratio between the mean-square error and the noise variance, at a given $R$ and $\sigma^2$. As a succinct characterization of robustness, we focus particular attention on the worst-case noise sensitivity:
\begin{definition}
The worst-case noise sensitivity of optimal encoding is defined as
\begin{equation}
	\zeta^*(X,R) = 	\sup_{\sigma^2 > 0} \frac{\Dstar(X, R,\sigma^2)}{\sigma^2}.
	\label{eq:zstar}
\end{equation}
For linear encoding, $\zLstar(X,R)$ and $\zL(X,R)$ are analogously defined with $\Dstar$ in \prettyref{eq:zstar} replaced by $\DLstar$ and $\DL$, respectively.
	\label{def:wcns}
\end{definition}
\begin{remark}
	In the analysis of LASSO and the AMP algorithms \cite{maleki.noise.sens}, 
	the noise sensitivity is defined in a minimax fashion where a further supremum is taken over all input distributions that have an atom at zero of mass at least $1-\epsilon$.  
	In contrast, the sensitivity in \prettyref{def:wcns} is a Bayesian quantity where we fix the input distribution. Similar notion of sensitivity has been defined in \cite[Equation (49)]{RG12}. 
\end{remark}

The phase transition threshold of the noise sensitivity is defined as the minimal measurement rate $R$ such that the noise sensitivity is bounded for all noise variance \cite{maleki.noise.sens,noisyCS.poster}:




\begin{definition}
Define
	\begin{equation}
\Rstar(X) \triangleq \inf\sth{R > 0\colon \zeta^*(X,R) < \infty}.
	\label{eq:Rstar}
\end{equation}
For linear encoding, $\RLstar(X)$ and $\RL(X)$ are analogously defined with $\zstar$ in \prettyref{eq:Rstar} replaced by $\zLstar$ and $\zL$.
	\label{def:pt.noisy}
\end{definition}

By \prettyref{eq:rank}, the phase transition thresholds in \prettyref{def:pt.noisy} are ordered naturally as
\begin{equation}
	0 \leq \Rstar(X) \leq \Rstar_{\rm L}(X) \leq \RL(X) \leq 1,
	\label{eq:rankR}
\end{equation}
where the rightmost inequality is shown below (after \prettyref{thm:worstG}).

\begin{remark}
	In view of the convexity properties in \prettyref{thm:distortion.property}.\ref{prop2}, the three worst-case sensitivities in \prettyref{def:wcns} are all (extended real-valued) convex functions of $R$.
	\label{rmk:wcns.cvx}
\end{remark}

\begin{remark}
Alternatively, we can consider the \emph{asymptotic noise sensitivity} by replacing the supremum in \prettyref{eq:zstar}
with the limit as $\sigma^2 \to 0$, denoted by $\xstar, \xLstar$ and $\xL$ respectively. Asymptotic noise sensitivity characterizes the convergence rate of the reconstruction error as the noise variance vanishes. Since $\Dstar(X, R,\sigma^2)$ is always bounded above by $\var X = 1$, we have
\begin{equation}
\zstar(X, R) < \infty \Leftrightarrow \xstar(X, R) < \infty.
	\label{eq:suplim}
\end{equation}
Therefore $\Rstar(X)$ can be equivalently defined as the infimum of all rates $R > 0$, such that
\begin{equation}
	\Dstar(X, R,\sigma^2) = \Bigo{\sigma^2}, \quad \sigma^2 \to 0.
	\label{eq:Rstar.equiv}
\end{equation}
This equivalence also applies to $\DLstar$ and $\DL$. It should be noted that although finite worst-case noise sensitivity is equivalent to finite asymptotic noise sensitivity, the supremum in \prettyref{eq:suplim} need not be achieved as $\sigma^2 \to 0$. An example is given by the Gaussian input analyzed in \prettyref{sec:gaussian.jscc}.


\label{rmk:suplim}
\end{remark}

\subsection{Least-favorable input: Gaussian distribution}
\label{sec:gaussian.jscc}
In this section we compute the distortion-rate tradeoffs for the Gaussian input distribution. Although Gaussian input distribution is not directly relevant for compressed sensing due to its lack of sparsity, it is still interesting to investigate the distortion-rate tradeoff in the Gaussian case for the following reasons:
\begin{enumerate}
	\item As the least-favorable input distribution, Gaussian distribution simultaneously maximizes all three distortion-rate functions subject to the variance constraint and provides upper bounds for non-Gaussian inputs.
	\item Connections are made to classical joint-source-channel-coding problems in information theory about transmitting Gaussian sources over Gaussian channels and (sub)optimality of linear coding (\eg, \cite{Goblick65, Ziv70,GRV03}).
	\item It serves as an concrete illustration of the phenomenon of coincidence of all thresholds defined in Definitions \ref{def:Dstar} -- \ref{def:DL},
	which are fully generalized in \prettyref{sec:results.noisy} to the mixture model.
\end{enumerate}

\begin{theorem}
Let $\XGn \sim \calN(0,1)$. Then for any $R, \sigma^2$ and $X$ of unit variance,
\begin{align}
\Dstar(X, R,\sigma^2) \leq \Dstar(X_{\sf G}, R,\sigma^2) = & ~ \frac{1}{(1+ \sigma^{-2})^R}. \label{eq:Dstar.g} \\
\DLstar(X, R,\sigma^2) \leq \DLstar(X_{\sf G}, R,\sigma^2) = & ~ 1 - \frac{R}{\sigma^2 + \max\{1,R\}}
	\label{eq:DLstar.g}
	\\
\DL(X, R,\sigma^2) \leq	\DL(X_{\sf G}, R,\sigma^2) = & ~  \frac{1}{2} \left(1-R-\sigma ^2+\sqrt{(1-R)^2+2 (1+R) \sigma ^2+\sigma ^4}\right)
\label{eq:DL.g}
\end{align}
	\label{thm:worstG}
\end{theorem}
\begin{proof}
Since the Gaussian distribution maximizes the rate-distortion function pointwise under the variance constraint \cite[Theorem 4.3.3]{berger}, the inequality in \prettyref{eq:Dstar.g} follows from \prettyref{eq:Dstar.sepa}. 
For linear encoding, linear estimators are optimal for Gaussian inputs since the channel output and the input are jointly Gaussian, but suboptimal for non-Gaussian inputs. Moreover, the linear MMSE depends only on the input variance. Therefore the inequalities in \prettyref{eq:DLstar.g} and \prettyref{eq:DL.g} follow.
The distortion-rate functions of $\XGn$ are computed in \prettyref{app:gaussian.linear}.
\end{proof}


The Gaussian distortion-rate tradeoffs in \prettyref{eq:Dstar.g} -- \prettyref{eq:DL.g} are plotted in Figs. \ref{fig:plot} and \ref{fig:plot.snr}. 
We see that linear encoders are optimal for lossy encoding of Gaussian sources in Gaussian channels if and only if $R = 1$, \ie,
\begin{equation}
	D^*(\XGn,1,\sigma^2) = \DLstar(\XGn,1,\sigma^2),
\end{equation}
which is a well-known fact \cite{Goblick65,Ziv70}. As a result of \prettyref{eq:DL.g}, the rightmost inequality in \prettyref{eq:rankR} follows. 

\begin{figure}[htp]
	\centering
	\begin{overpic}[
	width=.9\columnwidth]{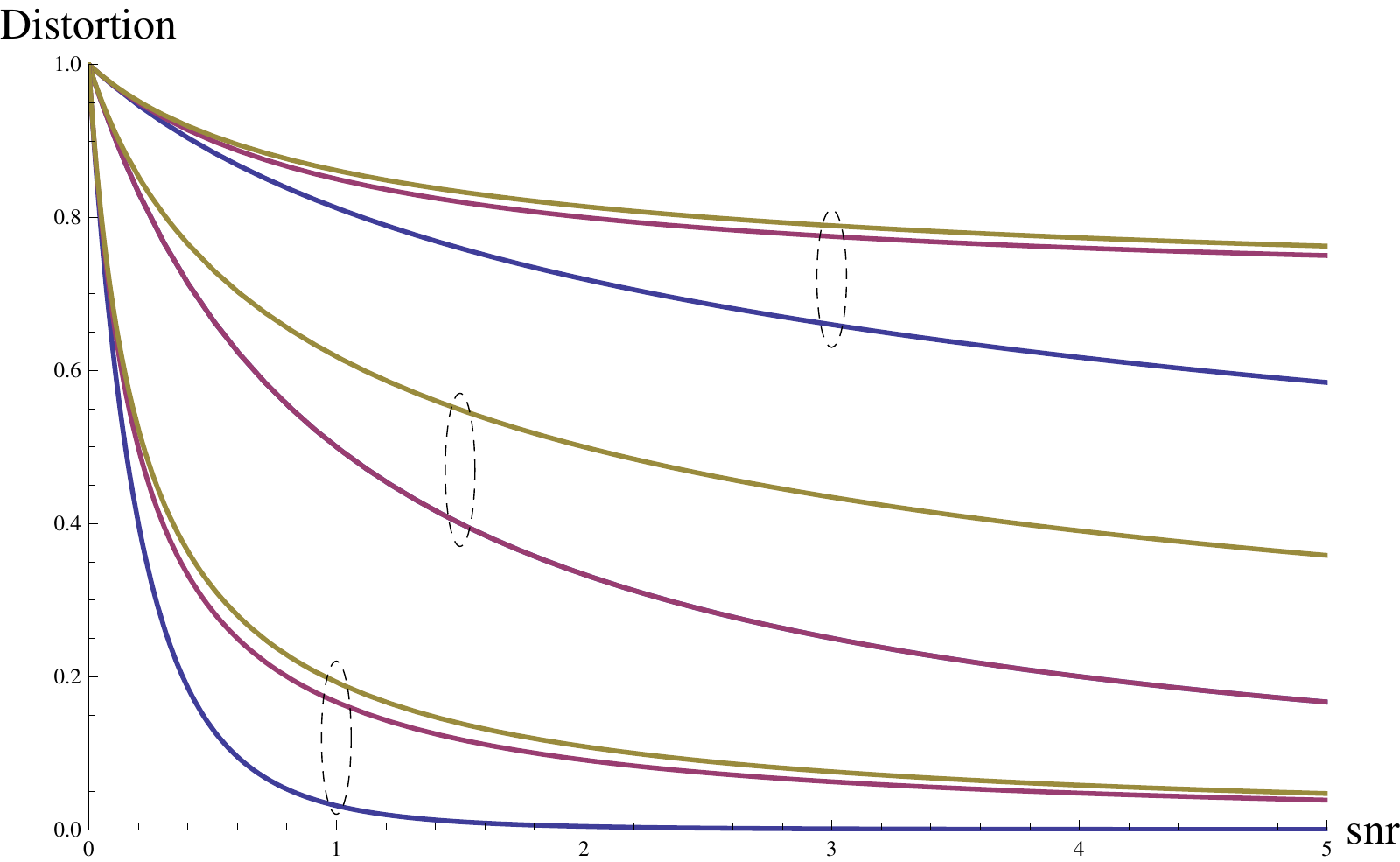}
\put(26,6){$R = 5$}
\put(35,26){$R = 1$}
\put(62,40){$R = 0.3$}
\end{overpic}
\caption{$\Dstar(\XGn, R,\sigma^2), \DLstar(\XGn, R,\sigma^2), \DL(\XGn, R,\sigma^2)$ against $\snr = \sigma^{-2}$.}
	\label{fig:plot}
\end{figure}

\begin{figure}[htp]
	\centering
	\begin{overpic}[
	width=.9\columnwidth]{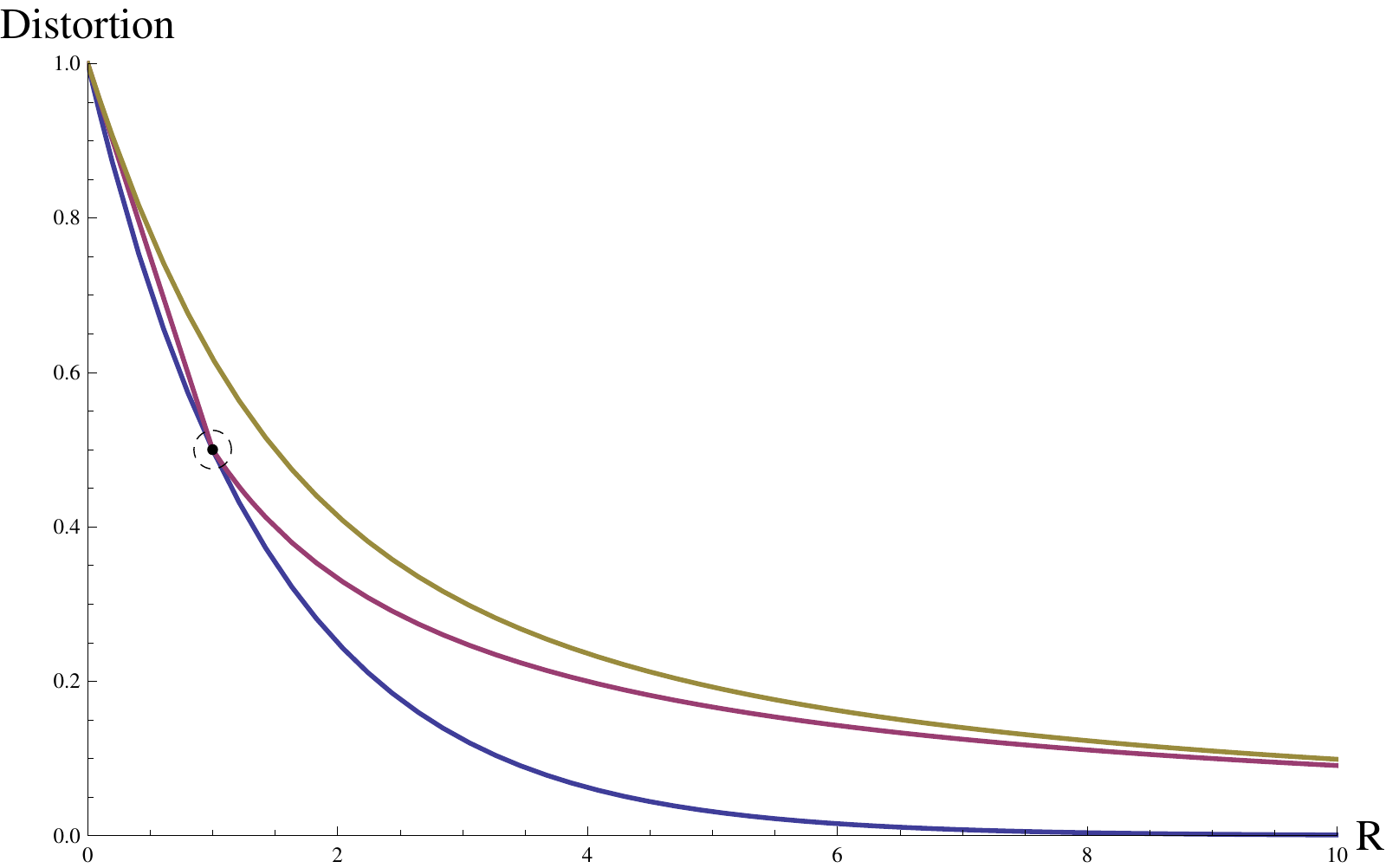}
\put(20,40){\vector(-1,-2){4.8}}
\put(20,40){Linear coding is optimal iff $R = 1$}
\put(55,20){\color{mathyellow}\vector(-1,-2){3.4}}
\put(55,20){\textcolor{mathyellow}{Random linear}}
\put(60,15){\color{mathred}\vector(-1,-2){2.3}}
\put(60,15){\textcolor{mathred}{Optimal linear}}
\put(75,10){\color{mathblue}\vector(-1,-2){3.7}}
\put(75,10){\textcolor{mathblue}{Optimal}}
\end{overpic}
\caption{$\Dstar(\XGn, R,\sigma^2), \DLstar(\XGn, R,\sigma^2), \DL(\XGn, R,\sigma^2)$ against $R$ when $\sigma^2 = 1$.}
	\label{fig:plot.snr}
\end{figure}

%

Next, using straightforward limits, we analyze the high-SNR asymptotics of \prettyref{eq:Dstar.g} -- \prettyref{eq:DL.g}. The smallest among the three, $\Dstar(\XGn, R, \sigma^2)$ vanishes polynomially in $\sigma^2$ according to
\begin{equation}
	\Dstar(\XGn, R,\sigma^2) = \sigma^{2R} + O(\sigma^{2R+2}), \quad \sigma^2 \to 0
	\label{eq:Dstar.g.asymp}
\end{equation}
regardless of how small $R > 0$ is. For linear encoding, we have
\begin{align}
\DLstar(X, R,\sigma^2) & ~ = \begin{cases}
	1 - R + R \sigma^2 + \Bigo{\sigma^4} & 0 \leq R < 1, \\
	\sigma^2 + O(\sigma^4)& R = 1, \\
	\frac{\sigma^2}{R}  + O(\sigma^4) & R > 1.
	\end{cases}
	\label{eq:DLstar.g.asymp} \\
	& \nonumber \\
\DL(\XGn, R,\sigma^2) & ~ = \begin{cases}
	1 - R + \frac{R}{1-R} \sigma^2 + O(\sigma^4) & 0 \leq R < 1, \\
	\sigma - \frac{\sigma^2}{2} + O(\sigma^3) & R = 1, \\
	\frac{\sigma^2}{R-1}  + O(\sigma^4) & R > 1.
	\end{cases}
	\label{eq:DL.g.asymp}
\end{align}
The weak-noise behavior of $\DLstar$ and $\DL$ are compared in different regimes of measurement rates:
\begin{itemize}
\item $0 \leq R < 1$:
both $\DLstar$ and $\DL$	converge to $1-R > 0$. This is an intuitive result, because even in the absence of noise, the orthogonal projection of the input vector onto the nullspace of the sensing matrix cannot be recovered, which contributes a total mean-square error of $(1-R)n$;
	Moreover, $\DL$ has strictly worse second-order asymptotics than $\DLstar$, especially when $R$ is close to 1.
\item $R = 1$: $\DL = \sigma(1+o(1))$ is much worse than $\DLstar = \sigma^2(1+o(1))$, which is achieved by choosing the encoding matrix to be identity. In fact, with nonnegligible probability, the optimal estimator that attains \prettyref{eq:DL.g}	blows up the noise power when inverting the random matrix;
\item $R > 1$:
	both $\DLstar$ and $\DL$ behave according to $\Theta(\sigma^2)$, but the scaling constant of $\DLstar$ is strictly worse, especially when $R$ is close to 1.
\end{itemize}

The foregoing high-SNR analysis shows that the average performance of random sensing matrices with \iid~entries is much worse than that of optimal sensing matrices, except if $R \ll 1$ or $R \gg 1$. Although this conclusion stems from the high-SNR asymptotics, we test it with several numerical results.
\prettyref{fig:plot} ($R = 0.3$ and 5) and \prettyref{fig:plot.snr} ($\sigma^2 = 1$) illustrate that the superiority of optimal sensing matrices carries over to the regime of non-vanishing $\sigma^2$. However, as we will see, randomly selected matrices are as good as the optimal matrices (and in fact, optimal nonlinear encoders) as far as the phase transition threshold of the worst-case noise sensitivity is concerned.


From \prettyref{eq:Dstar.g.asymp} and \prettyref{eq:DL.g.asymp}, we observe that both $\DLstar$ and $\DL$ exhibit a sharp phase transition near the critical rate $R = 1$:
	\begin{align}
\lim_{\sigma^2 \to 0} \DLstar(X, R,\sigma^2)	
= & ~ \lim_{\sigma^2 \to 0} \DL(X, R,\sigma^2)	\\
= & ~ (1-R)^+. 
\end{align}
where $x^+ \triangleq \max\{0,x\}$. Moreover, from \prettyref{eq:Dstar.g} -- \prettyref{eq:DL.g} we obtain the worst-case and asymptotic noise sensitivity functions for the Gaussian input as follows:
\begin{equation}
	\zstar(\XGn,R) = 
\begin{cases}	
\exp\pth{ - R h\left(\frac{1}{R}\right) }
& R \geq 1\\
\infty & R < 1
\end{cases}
,
\end{equation}
\begin{equation}
\xstar(\XGn,R) = 
\begin{cases}	
0 & R > 1\\
1 & R = 1\\
\infty & R < 1
\end{cases}
	\label{eq:zstar.g}
\end{equation}
and
\begin{equation}
	\zLstar(\XGn,R) = \xLstar(\XGn,R) = \begin{cases}	
\frac{1}{R} & R \geq 1\\
\infty & R < 1
\end{cases}
,
	\label{eq:zLstar.g}
\end{equation}
\begin{equation}
	\zL(\XGn,R) = \xL(\XGn,R) = \begin{cases}	
\frac{1}{R-1} & R > 1\\
\infty & R \leq 1
\end{cases}
	\label{eq:zL.g}
\end{equation}
The worst-case noise sensitivity functions are plotted in \prettyref{fig:worstsens.g} against the measurement rate $R$. Note that \prettyref{eq:zstar.g} provides an example for \prettyref{rmk:suplim}: for Gaussian input and $R > 1$, the asymptotic noise sensitivity for optimal coding is zero, while the worst-case noise sensitivity is always strictly positive.

\begin{figure}[ht]
	\centering
	\begin{overpic}[
	width=.9\columnwidth]{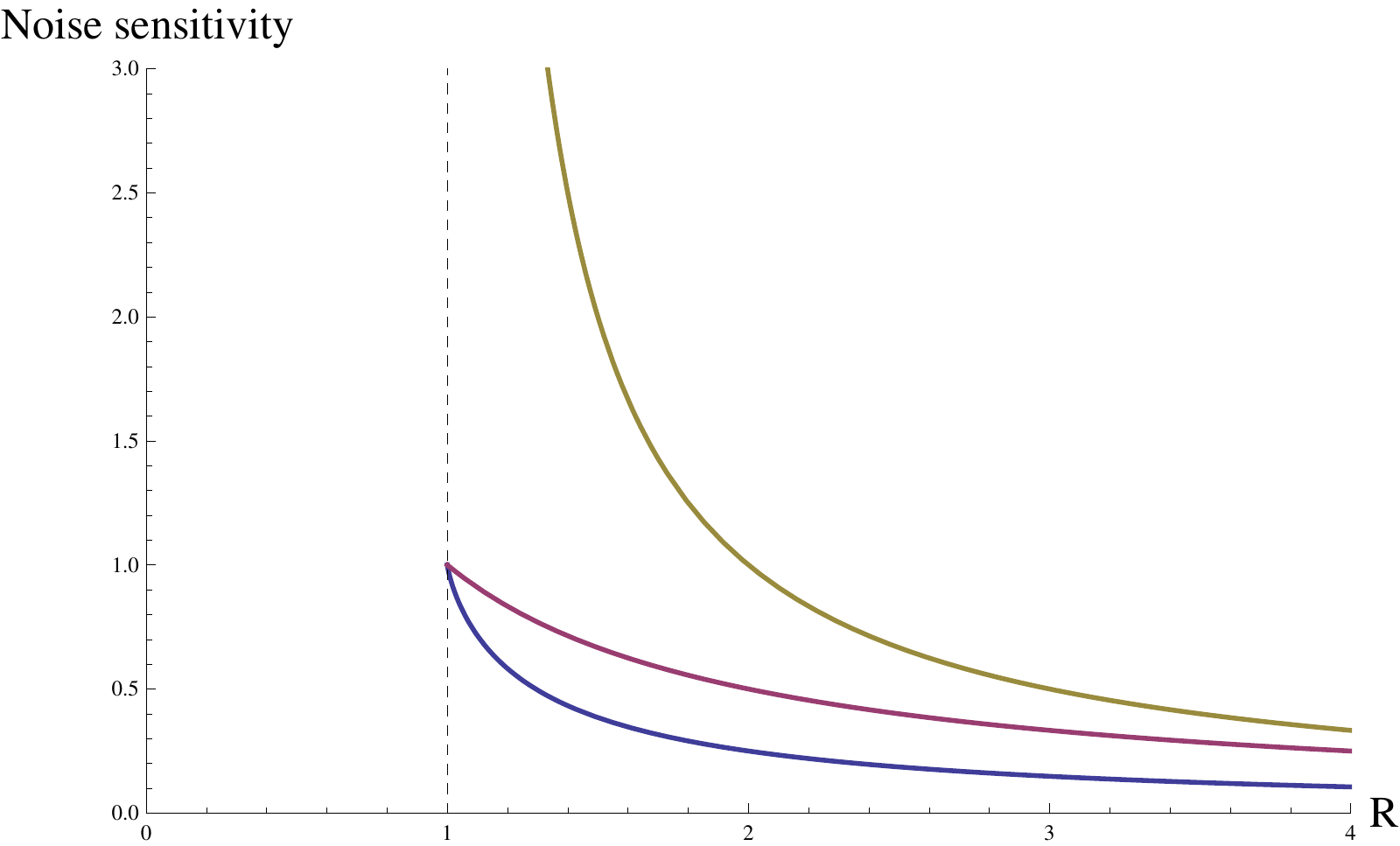}
	\put(41,22){\color{mathred} \vector(-1,-2){3}}
\put(41,22){\color{mathred} $\zLstar = \frac{1}{R}$}
\put(50,34){\color{mathyellow} \vector(-1,-2){3}}
\put(50,34){\color{mathyellow} $\zL = \frac{1}{R-1}$}
\put(82,17){\color{mathblue} \vector(-1,-2){6}}
\put(82,17){\color{mathblue} $\zstar = \frac{(R-1)^{R-1}}{R^R}$}
\put(20,20){{\large $+\infty$}}
\put(15,30){\textbf{unstable}}
\put(70,30){\textbf{stable}}
\end{overpic}
	\caption{Worst-case noise sensitivity $\zstar, \zLstar$ and $\zL$ for the Gaussian input, which all become infinity when $R < 1$ (the unstable regime).}
	\label{fig:worstsens.g}
\end{figure}

In view of \prettyref{eq:zstar.g} -- \prettyref{eq:zL.g}, the phase-transition thresholds in the Gaussian signal case are:
\begin{equation}
	\Rstar(\XGn) = \RLstar(\XGn) = \RL(\XGn) = 1.
	\label{eq:threshold.g}
\end{equation}
The equality of the three phase-transition thresholds turns out to hold well beyond the Gaussian signal model. In the next subsection, we formulate and prove the existence of the phase thresholds for all three distortion-rate functions and discrete-continuous mixtures, which turn out to be equal to the information dimension of the input distribution.

\subsection{Non-Gaussian inputs}
	\label{sec:results.noisy}
	This subsection contains our main results, which show that the phase transition thresholds are equal to the information dimension of the input, under rather general conditions. Therefore, the optimality of random sensing matrices in terms of the worst-case sensitivity observed in \prettyref{sec:gaussian.jscc} carries over well beyond the Gaussian case.
	Proofs are deferred to \prettyref{sec:pf.b}.
	
The phase transition threshold for \emph{optimal encoding} is given by the upper \emph{information dimension} of the input:
\begin{theorem}
For any $X$ that satisfies \prettyref{eq:dfinite}, 
	\begin{equation}
	\Rstar(X) = \od(X)
	\label{eq:RstarX}
\end{equation}
Moreover, if $P_X$ is a discrete-continuous mixture as in \prettyref{eq:dcmix}, then for any $R \geq \gamma$, as $\sigma \to 0$,
\begin{equation}
\Dstar(X,R,\sigma^2) =  \frac{\exp\pth{2 H(P_{\rm d}) \frac{1-\gamma}{\gamma} - 2 \calD(P_{\rm c})}}{(1-\gamma)^{\frac{2(1-\gamma)}{\gamma}}\gamma} \sigma^{\frac{2R}{\gamma}} 
 (1+o(1))
	\label{eq:Dstar.mix}
\end{equation}
where $\calD(\cdot)$ denotes the non-Gaussianness of a probability measure, defined as its relative entropy with respect to a Gaussian distribution with the same mean and variance.
Consequently, the asymptotic noise sensitivity of optimal encoding is 
\begin{equation}
	\xstar(X,R) = \begin{cases}
	\infty & R < \gamma \\	
	\frac{\exp\pth{2 H(P_{\rm d}) \frac{1-\gamma}{\gamma} - 2 \calD(P_{\rm c})}}{(1-\gamma)^{\frac{2(1-\gamma)}{\gamma}}\gamma}
& R = \gamma \\	
	0 & R > \gamma.
	\end{cases}
	\label{eq:xstar.mix}
\end{equation}
\label{thm:Rstar}
\end{theorem}

The next result shows that random linear encoders with \iid \emph{Gaussian} coefficients also achieve information dimension for any discrete-continuous mixtures, which, in view of \prettyref{thm:Rstar}, implies that, at least asymptotically, (random) linear encoding suffices for robust reconstruction as long as the input distribution contains no singular component.
\begin{theorem}
Assume that $X$ has a discrete-continuous mixed distribution as in \prettyref{eq:dcmix}, where the discrete component $P^{\rm d}$ has finite entropy. Then
	\begin{equation}
	\Rstar(X) = \Rstar_{\rm L}(X) = \RL(X) = \gamma. 
	\label{eq:noisy.dc}
\end{equation}
Moreover, 
\begin{enumerate}
	\item \prettyref{eq:noisy.dc} holds for any non-Gaussian noise distribution with finite non-Gaussianness.
	\item For any $R>\gamma$, the worst-case noise sensitivity of Gaussian sensing matrices is upper bounded by
	\begin{equation}
	\xL(X,R) \leq
	\frac{R^2}{(R-\gamma)^2} \pth{\frac{R}{\gamma}}^{\frac{2 \gamma}{R-\gamma}} \exp\pth{ \frac{2 H(P_{\rm d}) (1 - \gamma)+ 2 h(\gamma)}{R-\gamma} + 1}.
	\label{eq:xL.mix}
\end{equation}
\end{enumerate}
	\label{thm:noisy.dc}
\end{theorem}

\begin{remark}
The achievability proof of $\RL(X)$ is a direct application of \prettyref{thm:linear.mix}, where the Lipschitz decompressor in the noiseless case is used as a suboptimal estimator in the noisy case. The outline of the argument is as follows: suppose that we have obtained a sequence of linear encoders and $L_R$-Lipschitz decoders $\{(\bfA_n ,g_n)\}$ with rate $R$ and error probability $\epsilon_n \to 0$ as $n \to \infty$. Then
	\begin{equation}
	\expect{\big\|g_n(\bfA_n X^n + \sigma N^k) - X^n\big\|^2} \leq L^2(R) \sigma^2 \expect{\big\|N^k\big\|^2} + \epsilon_n  = k L^2(R) \sigma^2 \var N + \epsilon_n,
	\label{eq:lipbased}
\end{equation}
which implies that robust reconstruction is achievable at rate $R$ and the worst-case noise sensitivity is upper bounded by $L^2_R R$.

Notice that the above achievability approach applies to any noise with finite variance, without requiring that the noise be additive, memoryless or that it have a density.
 In contrast, replica-based results rely crucially on the fact that the additive noise is memoryless Gaussian. Of course, in order for the converse (via $\RLstar$) to hold, the non-Gaussian noise needs to have finite non-Gaussianness. The disadvantage of this approach is that currently it lacks an explicit construction because the extendability of Lipschitz functions (Kirszbraun's theorem) is only an existence result which relies on the Hausdorff maximal principle \cite[Theorem 1.31, p. 21]{NLFA.Schwartz}, which is equivalent to the axiom of choice. On Euclidean spaces it is possible to obtain an explicit construction by applying the results in \cite{Brehm81,AT08} to a countable dense subset of the domain. However, such a construction is far from being practical.

	
\end{remark}



\begin{remark}
We emphasize the following ``universality'' aspects of \prettyref{thm:noisy.dc}:
	\begin{itemize}
	\item Gaussian random sensing matrices achieve the optimal transition threshold for any discrete-continuous mixture, as long as it is known at the decoder;
	\item The fundamental limit depends on the input statistics only through the weight on the analog component, regardless of the specific discrete and continuous components. In the conventional sparsity model \prettyref{eq:sparseP} where $P_X$ is the mixture of an absolutely continuous distribution and a mass of $1-\gamma$ at 0, the fundamental limit is $\gamma$;
	\item The suboptimal estimator used in the achievability proof comes from the noiseless Lipschitz decoder, which does not depend on the noise distribution, or even its variance;
	\item The conclusion holds for non-Gaussian noise as long as it has finite non-Gaussianness.
\end{itemize}
\end{remark}

\subsection{Results relying on replica heuristics}
	\label{sec:replica}
Based on the statistical-physics approach in \cite{tanaka.2004,guo.cdma}, the decoupling principle results in \cite{guo.cdma} were imported into the compressed sensing setting in \cite{dongning.noisyCS} to \emph{postulate} the following formula for $\DL(X,R,\sigma^2)$. Note that this result is based on replica heuristics currently lacking a rigorous justification.
\begin{RSP}[{\cite[Corollary 1, p.5]{dongning.noisyCS}}]
\begin{equation}
	\DL(X, R,\sigma^2) = 
	\mmse(X, \eta R \,\sigma^{-2}),
	\label{eq:om.mmse}
\end{equation}
where $0 < \eta < 1$ satisfies the following equation \cite[(12) -- (13), pp. 4 -- 5]{dongning.noisyCS}:\footnote{In the notation of \cite[(12)]{dongning.noisyCS}, $\gamma$ and $\epsilon \mu$ correspond to $R \sigma^{-2}$ and $R$ in our formulation.}
\begin{equation}
	\frac{1}{\eta} = 1 + \frac{1}{\sigma^2} \mmse(X, \eta R \,\sigma^{-2}).
	\label{eq:eta}
\end{equation}
When \prettyref{eq:eta} has more than one solution, $\eta$ is chosen to minimize the free energy
\begin{equation}
	I(X; \sqrt{\eta R \sigma^{-2}} X + N) + \frac{R}{2} (\eta - 1 - \log \eta).
	\label{eq:eta.min}
\end{equation}
\label{clm:replica}
\end{RSP}

In view of the the I-MMSE relationship \cite{guo.immse}, the solutions to \prettyref{eq:eta} are precisely the stationary points of the free energy \prettyref{eq:eta.min} as a function of $\eta$. In fact it is possible for \prettyref{eq:om.mmse} to have arbitrarily many solutions. For an explicit example, see \prettyref{rmk:cantor} in \prettyref{sec:pf.b}.

Note that the solution in \prettyref{eq:om.mmse} does \emph{not} depend on the distribution of the random measurement matrix $\bsA$, as long as its entries are \iid with zero mean and variance $\frac{1}{n}$. Therefore it is possible to employ a random sparse measurement matrix so that each encoding operation involves only a relatively few signal components, for example, 
\begin{equation}
A_{ij} \sim \frac{p}{2} \delta_{\frac{-1}{\sqrt{p n}}} + (1-p) \delta_0 + \frac{p}{2} \delta_{\frac{1}{\sqrt{p n}}}
	\label{eq:sparseA}
\end{equation}
for some $0 < p < 1$. In fact, in the special case of $p=\frac{\log n}{n}$, the replica symmetry postulate can be rigorously proved \cite[Sec. IV]{dongning.noisyCS} (see also \cite{GW08,Montanari08}).

Assuming the validity of the replica symmetry postulate, it can be shown that the phase transition threshold for random linear encoding is always sandwiched between the lower and the upper \emph{MMSE dimension} of the input. The relationship between the MMSE dimension and the information dimension in \prettyref{eq:mmse.renyi} plays a key role in analyzing the minimizer of the free energy \prettyref{eq:eta.min}.\footnote{It can be shown that in the limit of $\sigma^2 \to 0$, the minimizer of \prettyref{eq:eta.min} when $R > \oscrD(X)$ and $R < \uscrD(X)$ corresponds to the largest and the smallest root of the fixed-point equation \prettyref{eq:om.mmse} respectively.}
\begin{theorem}
Assume that the replica symmetry postulate holds for $X$. Then for any \iid random measurement matrix $\bsA$ whose entries have zero mean and variance $\frac{1}{n}$,
	\begin{equation}
	\uscrD(X) \leq \RL(X) \leq \oscrD(X).
	\label{eq:RLstarX.bd}
\end{equation}
Therefore if $\scrD(X)$ exists, we have
\begin{equation}
	\RL(X) = \scrD(X) = d(X),
	\label{eq:RLstarX}
\end{equation}
and in addition, 
\begin{equation}
	\DL(X,R,\sigma^2) = \frac{d(X)}{R - d(X)} \sigma^2 (1	+ o(1)).
	\label{eq:sens.mix}
\end{equation}
\label{thm:RLstar}
\end{theorem}

The general result in \prettyref{thm:RLstar} holds for any input distribution but relies on the conjectured validity of the replica symmetry postulate. For the special case of discrete-continuous mixtures in \prettyref{eq:dcmix}, in view of \prettyref{thm:Dmix}, \prettyref{thm:RLstar} predicts (with the caveat of the validity of the replica symmetry postulate) that the phase-transition threshold for \emph{Gaussian} sensing matrices is $\gamma$, which agrees with the rigorously proven result in \prettyref{thm:noisy.dc}. Therefore, the only added benefit of \prettyref{thm:RLstar} is to allow singular components in the input distribution.

\begin{remark}
In statistical physics, the phase transition near the threshold often behaves according to a power law with certain universal exponent, known as the \emph{critical exponent} \cite[Chapter 3]{Stanley71}. According to \prettyref{eq:sens.mix}, as the measurement rate $R$ approaches the fundamental limit $d(X)$, the replica method suggests that the optimal noise sensitivity blows up according as the power law $\frac{1}{R-d(X)}$, where the unit exponent holds universally for all mixture distributions. It remains a open question whether this power law behavior can be rigorously proven and whether the optimal exponent is one. 
Note that by using the Lipschitz extension scheme in the proof \prettyref{thm:noisy.dc}, we can achieve the noise sensitivity in \prettyref{eq:xL.mix}, which blows up exponentially as the $R - d(X)$ vanishes and is likely to highly suboptimal.

	\label{rmk:critexp}
\end{remark}

\begin{remark}
	In fact, the proof of \prettyref{thm:RLstar} shows that the converse part (left inequality) of \prettyref{eq:RLstarX.bd} holds in a much stronger sense: as long as there is no residual error in the weak-noise limit, that is, if $\DL(X,R,\sigma^2) = o(1)$ as $\sigma^2 \to 0$, then $R \geq \uscrD(X)$ must hold. Therefore, the converse part of \prettyref{thm:RLstar} still holds even if we weaken the right-hand side of \prettyref{eq:Rstar.equiv} from $O(\sigma^2)$ to $o(1)$.
\end{remark}

\begin{remark}
Assume the validity of the replica symmetry postulate. Combining \prettyref{thm:Rstar}, \prettyref{thm:RLstar} and \prettyref{eq:rankR} gives an \emph{operational} proof for 
	$\od(X) \leq \oscrD(X)$,
the fourth inequality in \prettyref{eq:mmse.renyi}, which has been proven analytically in \cite[Theorem 8]{mmse.dim.IT}. 
%
\end{remark}

\section{Comparisons to LASSO and AMP algorithms}
	\label{sec:compare}
Widely popular in the compressed sensing literature, the LASSO \cite{lasso,CDS99} and the approximate message passing (AMP) algorithms \cite{maleki.pnas} are low-complexity reconstruction procedures, which are originally obtained as solutions to the conventional minimax setup in compressed sensing. In this section, we compare the phase transition thresholds of LASSO and AMP achieved in the Bayesian setting to the optimal thresholds derived in Sections \ref{sec:noiseless} -- \ref{sec:noisy}. Similar Bayesian analysis has been performed in \cite{donoho.tanner.hypercube,RFG12,maleki.pnas,bayati.lasso,complex.lasso}.

	\subsection{Signal models}
	\label{sec:model5}
	The following three families of input distributions are considered \cite[p. 18915]{maleki.pnas}, indexed by $\chi = \pm, +$ and $\square$ respectively, which all belong to the family of input distributions of the mixture form in \prettyref{eq:dcmix}:
\begin{description}
	\item[$\pm$]: sparse signals \prettyref{eq:sparseP};
	\item[$+$]: sparse non-negative signals \prettyref{eq:sparseP} with the continuous component $P_{\rm c}$ supported on $\reals_+$.
	\item[$\square$]: simple signals (\prettyref{sec:model}) \cite[Section 5.2, p. 540]{donoho.tanner.hypercube}
\begin{equation}
P = (1-\gamma) \Big(\frac{1}{2} \delta_{0} + \frac{1}{2} \delta_{1}\Big)  + \gamma \, P_{\rm c}
	\label{eq:simpleP}
\end{equation}
where $P_{\rm c}$ is some absolutely continuous distribution supported on the unit interval.
\end{description}

\subsection{Noiseless measurements}
\label{sec:comp.a}
In the noiseless case, we consider 
linear programming (LP) 
decoders and the AMP decoder \cite{maleki.pnas} and the phase transition threshold of error probability. Phase transitions of greedy reconstruction algorithms have been analyzed in \cite{BCTT10}, which derived upper bounds (achievability results) for the transition threshold of measurement rate. We focus our comparison on algorithms whose phase transition thresholds are known exactly.

The following LP decoders are tailored to the three input distributions $\chi = \pm, +$ and $\square$ respectively (see Equations (P1), (LP) and (Feas) in \cite[Section I]{donoho.tanner.ieee10}):
	\begin{align}
g_{\scriptscriptstyle \pm}(y) = & ~ \arg\min\{ \lnorm{x}{1}: x \in \reals^n, \bsA x = y\}, \label{eq:l1.pm} \\
g_{\scriptscriptstyle +}(y) = & ~ \arg\min\{ \lnorm{x}{1}: x \in \reals^n_+, \bsA x = y\},
	\label{eq:l1.p} \\
g_{\scriptscriptstyle \square}(y) = & ~ \{ x: x \in [0,1]^n, \bsA x = y\}.	
	\label{eq:l1.sq}
\end{align}
For sparse signals, \prettyref{eq:l1.pm} -- \prettyref{eq:l1.p} are based on $\ell_1$-minimization (also known as Basis Pursuit \cite{CDS99}, which is the noiseless limit of LASSO defined in \prettyref{sec:comp.b}), while for simple signals, the decoder \prettyref{eq:l1.sq} solves an LP feasibility problem. In general the decoders in \prettyref{eq:l1.pm} -- \prettyref{eq:l1.sq} output a list of vectors upon receiving the measurement. The reconstruction is successful if and only if the output list contains only the true vector. The error probability is thus defined as $\prob{g_{\chi}(\bsA X^n) \neq \{X^n\}}$, evaluated with respect to the product measure $(P_X)^{n} \times P_{\bsA}$.

The phase transition thresholds of the reconstruction error probability for decoders \prettyref{eq:l1.pm} -- \prettyref{eq:l1.sq} are derived in \cite{donoho.tanner.2009} using combinatorial geometry. For sparse signals and $\ell_1$-minimization decoders \prettyref{eq:l1.pm} -- \prettyref{eq:l1.p}, the expressions of the corresponding thresholds $\sfR_{\scriptscriptstyle \pm}(\gamma)$ and $\sfR_{\scriptscriptstyle +}(\gamma)$ are quite involved, given implicitly in \cite[Definition 2.3]{donoho.tanner.2009}. As observed in \cite[Finding 1]{maleki.pnas}, 
$\sfR_{\scriptscriptstyle \pm}(\gamma)$ and $\sfR_{\scriptscriptstyle +}(\gamma)$ agree numerically with the following expressions:\footnote{In the series of papers \cite{donoho.tanner.hypercube, donoho.tanner.ieee10, maleki.pnas, maleki.noise.sens}, the phase diagrams are parameterized by $(\rho, \delta)$, where $\delta = R$ is the measurement rate and $\rho = \frac{\gamma}{R}$ is the ratio between the sparsity and rate. In this paper, the parameterization $(\gamma, R)$ is used instead. The ratio $\frac{\gamma}{\sfR_{\chi}(\gamma)}$ is denoted by $\rho(\gamma; \chi)$ in \cite{maleki.pnas}. The same parametrization is also used in \cite{DJoM11}.}
\begin{align}
\sfR_{\scriptscriptstyle \pm}(\gamma) & ~ = \min_{\alpha \geq 0} \gamma(1+\alpha^2) + 2(1-\gamma)((1+\alpha^2) \Phi(-\alpha)-\alpha \varphi(\alpha)) \label{eq:rp} \\
\sfR_{\scriptscriptstyle +}(\gamma) & ~ = \min_{\alpha \geq 0} \gamma(1+\alpha^2) + (1-\gamma)((1+\alpha^2) \Phi(-\alpha)-\alpha \varphi(\alpha)) \label{eq:rpm}
\end{align}
which is now rigorously established in view of the results in \cite{BM11}.
For simple signals, the phase transition threshold is proved to be \cite[Theorem 1.1]{donoho.tanner.hypercube}
\begin{equation}
	\sfR_{\scriptscriptstyle \square}(\gamma) = \frac{\gamma+1}{2}.
	\label{eq:rq}
\end{equation}
Moreover, substantial numerical evidence in \cite{maleki.pnas} suggests that the phase transition thresholds for the AMP decoder coincide with the LP thresholds for all three input distributions. The suboptimal thresholds obtained from \prettyref{eq:rp} -- \prettyref{eq:rq} are
 plotted in \prettyref{fig:dmm} along with the optimal threshold obtained from \prettyref{thm:linear.mix} which is $\gamma$.\footnote{A similar comparison between the suboptimal threshold $\sfR_{\scriptscriptstyle \pm}(\gamma)$ and the optimal threshold $\gamma$ has been provided in \cite[Fig. 2(a)]{KWT10} based on a replica-heuristic calculation.}
In the gray area below the diagonal in the $(\gamma,R)$-phase diagram, any sequence of sensing matrices and decompressors will fail to reconstruct the true signal with probability that tends to one. Moreover, we observe that the LP and AMP decoders are severely suboptimal unless $\gamma$ is close to one. 
\begin{figure}[ht]
	\centering
	\begin{overpic}[
	width=.6\columnwidth]{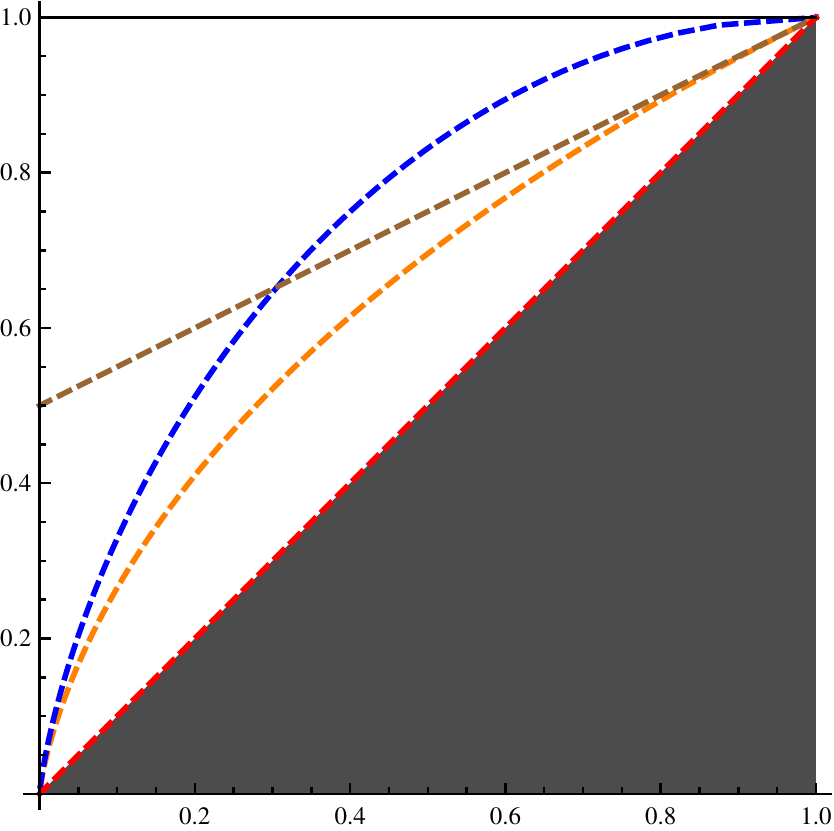}
\put(19,50){\textcolor{blue}{$\pm$}}
\put(25,40){\textcolor{orange}{$+$}}
\put(15,59){\textcolor{brown}{$\square$}}
\put(34,36){\turnbox{45}{{\small Optimal threshold}}}
	\put(50,0){$\gamma$}
  \put(-5,60){$R$}
\end{overpic}
\caption{Suboptimal thresholds \prettyref{eq:rp} -- \prettyref{eq:rq} obtained with LASSO and AMP v.s. optimal threshold for the three signal models in \prettyref{sec:model5}.}
  \label{fig:dmm}
\end{figure}

In the highly sparse regime which is most relevant to compressed sensing problems, it follows from \cite[Theorem 3]{donoho.tanner.ieee10} that for sparse signals ($\chi=\pm$ or $+$),
	\begin{equation}
\sfR_{\chi}(\gamma)
	= 2 \gamma \log_{\eexp} \frac{1}{\gamma} (1+o(1)), \text{ as } \gamma \to 0,
	\label{eq:l1.verysparse}
\end{equation}
which implies that $\sfR_{\chi}$ has infinite slope at $\gamma = 0$.
Therefore when $\gamma \ll 1$, the $\ell_1$ and AMP decoders require on the order of $2 s \log_{\eexp} \frac{n}{s}$ measurements to successfully recover the unknown vector, whose number of nonzero components is denoted by $s$. In contrast, $s$ measurements suffice when using an optimal decoder (or $\ell_0$-minimization decoder). The LP or AMP decoders are also highly suboptimal for simple signals, since $\sfR_{\scriptscriptstyle \square}(\gamma)$ converges to $\frac{1}{2}$ instead of zero as $\gamma \to 0$. This suboptimality is due to the fact that the LP feasibility decoder \prettyref{eq:l1.sq} simply finds any $x^n$ in the hypercube $[0,1]^n$ that is compatible with the linear measurements. Such a decoding strategy does not enforce the typical discrete structure of the signal, since most of the entries saturate at 0 or 1 equiprobably.
Alternatively, the following decoder achieves the optimal $\gamma$: define
\[
\sfT(x^n) = \frac{1}{n} \sum_{i=1}^n \pth{\indc{x_i = 0}, \indc{x_i \notin \{0,1\}}, \indc{x_i=1}}.
\]
The decoder outputs the solution to $\bsA x^n = y^k$ such that $\sfT(x^n)$ is closest to $\pth{\frac{1-\gamma}{2}, \gamma, \frac{1-\gamma}{2}}$ (in total variation distance for example).


\subsection{Noisy measurements}
\label{sec:comp.b}
In the noisy case, we consider the AMP decoder \cite{maleki.noise.sens} and the $\ell_1$-penalized least-squares (\ie LASSO) decoder \cite{lasso}:
\begin{equation}
	\tilde{g}(y, \bsA; \lambda) = \argmin_{x\in \reals^n} \frac{1}{2}\lnorm{y  -\bsA x}{2}^2 + \lambda \lnorm{x}{1},
	\label{eq:lasso}
\end{equation}
where $\lambda > 0$ is a regularization parameter. Note that in the limit of $\lambda \to 0$, LASSO reduces to the $\ell_1$-minimization decoder defined in \prettyref{eq:l1.pm}.
For Gaussian sensing matrices and Gaussian observation noise, the asymptotic mean-square error achieved by LASSO for a fixed $\lambda$ 
\begin{equation}
	D^{(\lambda)}(X,R,\sigma^2) 	
\triangleq	\lim_{n \to \infty} \frac{1}{n} \expect{\big\|X^n - \tilde{g}(\bsA X^n + \sigma  N^k; \lambda)\big\|^2}
\end{equation}
can be determined as a function of $P_X, \lambda$ and $\sigma$ by applying \cite[Corollary 1.6]{bayati.lasso}.\footnote{It should be noted that in \cite{maleki.noise.sens,bayati.lasso}, the entries of the sensing matrix is distributed according to $\calN(0,\frac{1}{k})$ (column normalization). While in the present paper the sensing matrix has $\calN(0,\frac{1}{n})$ entries  (row normalization) in order for the encoded signal to have unit average power. Therefore the expression in \prettyref{eq:sens.lasso} is equal to that in \cite[Equation (1.9)]{maleki.noise.sens} divided by the measurement $R$.}
In \prettyref{app:lasso}, we show that
for any $X$ distributed according to the mixture 
\begin{equation}
P_X = (1-\gamma)\delta_0+\gamma Q,	
	\label{eq:PXQ}
\end{equation}
where $Q$ is an arbitrary probability measure with no mass at zero,
the asymptotic
noise sensitivity of LASSO with optimized $\lambda$ is given by the following equation:
	\begin{align}
\tilde{\xi}(X,R) 
\triangleq  & ~ \inf_{\lambda} \lim_{\sigma^2 \to 0}\frac{D^{(\lambda)}(X,R,\sigma^2)}{\sigma^2} \\
= & ~ 	\begin{cases}
	\frac{\sfR_{\scriptscriptstyle \pm}(\gamma)}{R - \sfR_{\scriptscriptstyle \pm}(\gamma)} & R > \sfR_{\scriptscriptstyle \pm}(\gamma)\\
	\infty & R \leq  \sfR_{\scriptscriptstyle \pm}(\gamma)
	\end{cases}
	\label{eq:sens.lasso}
\end{align}
where $\sfR_{\scriptscriptstyle \pm}(\gamma)$ is given in \prettyref{eq:rpm}. By the same reasoning in \prettyref{rmk:suplim}, the worst-case noise sensitivity of LASSO is finite if and only if $R > \sfR_{\scriptscriptstyle \pm}(\gamma)$. 
Note that \prettyref{eq:sens.lasso} does not depend on $Q$ as long as $Q(\{0\})=0$. Therefore $\sfR_{\scriptscriptstyle \pm}(\gamma)$
also coincides with the phase transition threshold in a minimax sense, obtained in \cite[Proposition 3.1(1.a)]{maleki.noise.sens} by considering the lease favorable $Q$. 
Analogously, the LASSO decoder \prettyref{eq:lasso} can be adapted to other signal structures (see for example \cite[Sec. VI-A]{maleki.noise.sens}), resulting in the phase-transition threshold $\sfR_{\scriptscriptstyle +}(\gamma)$ and $\sfR_{\scriptscriptstyle \square}(\gamma)$ for sparse positive and simple signals, given by \prettyref{eq:rpm} and \prettyref{eq:rq}, respectively. Furthermore, these thresholds also apply to the AMP algorithm \cite{Montanari11}. 

Next, focusing on sparse signals, we compare the performance of LASSO and AMP algorithms to the optimum. In view of \prettyref{eq:sens.lasso}, the phase transition thresholds of noise sensitivity for the LASSO and AMP decoder are both $\sfR_{\scriptscriptstyle \pm}(\gamma)$ for any $X$ distributed according to \prettyref{eq:PXQ}. We discuss the following two special cases:
\begin{enumerate}
	\item $Q$ is absolutely continuous, or alternatively, $P_X$ is a discrete-continuous mixture given in \prettyref{eq:sparseP}. The optimal phase transition threshold is $\gamma$ as a consequence of \prettyref{thm:noisy.dc}. Therefore the phase transition boundaries are identical to \prettyref{fig:dmm} and the same observation in \prettyref{sec:comp.a} applies.
	\item $Q$ is discrete with no mass at zero, \eg, $Q = \frac{1}{2}(\delta_1+\delta_{-1})$. Since $P_X$ is discrete with zero information dimension, the optimal phase transition threshold is equal to zero, while $\sfR_{\scriptscriptstyle \pm}(\gamma)$ still applies to LASSO and AMP.
\end{enumerate}
For sparse signals of the form \prettyref{eq:sparseP} with $\gamma = 0.1$, \prettyref{fig:noisy.sparse} compares those expressions for the asymptotic noise sensitivity of LASSO (and AMP) algorithm to the optimal noise sensitivity predicted by \prettyref{thm:RLstar} based on replica heuristics. Note that the phase transition threshold of LASSO is approximately 3.3 times the optimal.



\begin{figure}[!t]
	\centering
\begin{overpic}[width=.9\columnwidth]{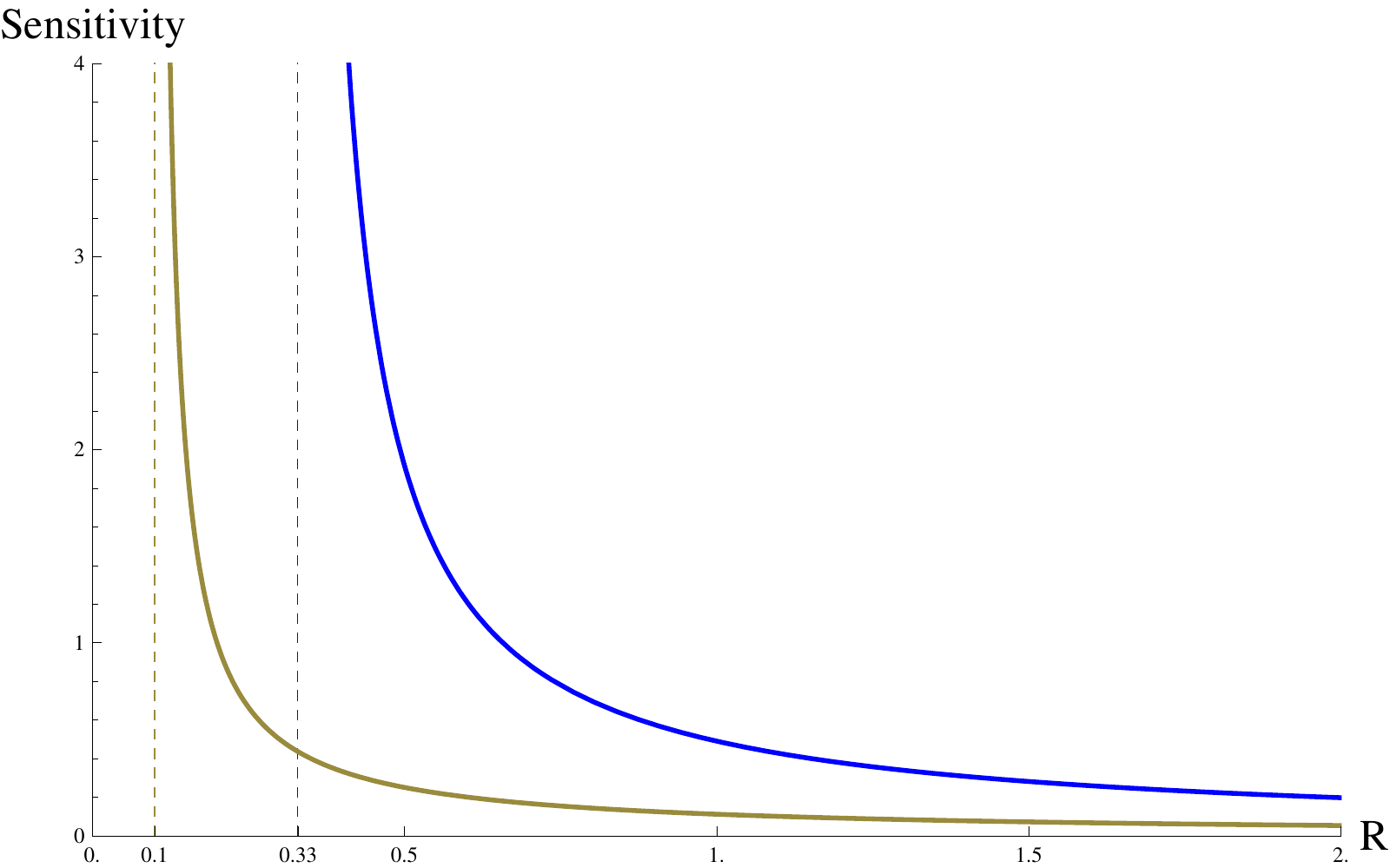}
\put(45,20){\color{blue} \vector(-1,-2){3.8}}
\put(45,20){\color{blue} LASSO decoder}
\put(75,10){\color{mathyellow} \vector(-1,-2){3.3}}
\put(75,10){\color{mathyellow} optimal decoder}
\end{overpic}
\caption{Asymptotic noise sensitivity  of the optimal decoder and the LASSO decoder exhibiting phase transitions: sparse signal model \prettyref{eq:sparseP} with $\gamma = 0.1$.}
	\label{fig:noisy.sparse}
\end{figure}

\section{Proofs}
	\label{sec:pf}
	\subsection{Auxiliary results}
	\label{sec:noiless.aux}
	
	We need the following large-deviations result on Gaussian random matrices.
%

\begin{lemma}
Let $\sigma_{\min}(\bsB_k)$ denote the smallest singular value of the $k \times m_k$ matrix $\bsB_k$ consisting of \iid Gaussian entries with zero mean and variance $\frac{1}{k}$. For any $t > 0$, denote 
\begin{equation}
F_{k,m_k}(t) \triangleq \prob{\sigma_{\min}(\bsB_k) \leq t}. 	
	\label{eq:Fkm}
\end{equation}
Suppose that $\frac{m_k}{k} \xrightarrow{k \to \infty} \alpha \in (0,1)$.
Then 
\begin{equation}
\liminf_{k \to \infty} \frac{1}{k} \log \frac{1}{F_{k,m_k}(t)} \geq \frac{(1-\alpha)}{2} \log \frac{(1-\alpha)^2}{\eexp t^2 } + \frac{\alpha}{2} \log \alpha.
	\label{eq:sigma1}
\end{equation}
	\label{lmm:edelman}
\end{lemma}
\begin{proof}
For brevity let $\bsH_k = \sqrt{k} \bsB_k$ and suppress the dependence of $m_k$ on $k$. Then $\bsH_k\trans\bsH_k$ is an $m \times m$ Gaussian Wishart matrix. The minimum eigenvalue of $\bsH_k\trans\bsH_k$ has a density, which admits the following upper bound \cite[Proposition 5.1, p. 553]{Edelman88}.
	\begin{equation}
	f_{\lambda_{\min}(\bsH_k\trans\bsH_k)}(x) \leq E_{k,m} x^{\frac{k-m-1}{2}} \eexp^{-\frac{x}{2}}, \quad x \geq 0,
	\label{eq:lambdamin}
\end{equation}
where 
\begin{equation}
E_{k,m} \triangleq  \frac{\sqrt{\pi} 2^{-\frac{k-m+1}{2}} \Gamma\pth{\frac{k+1}{2}}}{\Gamma\pth{\frac{m}{2}} \Gamma\pth{\frac{k-m+1}{2}} \Gamma\pth{\frac{k-m+2}{2}}}.
	\label{eq:Emn}
\end{equation}
Then
\begin{align}
\prob{\sigma_{\min}(\bsB_n) \leq t}
= & ~ \prob{\lambda_{\min}(\bsH_n\trans\bsH_n) \leq k t^2}\\
\leq & ~ E_{k,m} \int_{0}^{k t^2} x^{\frac{k-m-1}{2}} \eexp^{-\frac{x}{2}} \diff x\\
\leq & ~ \frac{2 k^{\frac{k-m+1}{2}} E_{k,m}}{k-m+1} t^{k-m+1}. \label{eq:gammas}
\end{align}
Applying Stirling's approximation to \prettyref{eq:gammas} yields \prettyref{eq:sigma1}.
\end{proof}

\begin{remark}
More general non-asymptotic upper bound on $\prob{\sigma_{\min}(\bsB_k) \leq t}$ is given in \cite[Theorem 1.1]{RV09}, which holds universally for all sub-Gaussian distributions. Note that the upper bound in \cite[Equation (1.10)]{RV09} is of the form 
\begin{equation}
\prob{\sigma_{\min}(\bsB_k) \leq t} \leq (c_1 t)^{c_2 k} + \exp(-c_3 k)	
	\label{eq:RV09}
\end{equation}
where $c_1,c_2,c_3$ are constants. The second term in \prettyref{eq:RV09} is due to the fact that the least singular value for discrete ensembles (\eg Rademacher) always has a mass at zero, which is exponentially small in $k$ but independent of $t$.
For Gaussian ensembles, however, we have $\sigma_{\min}(\bsB_k) > 0$ almost surely. Indeed, \prettyref{lmm:edelman} indicates that the second term in \prettyref{eq:RV09} can be dropped, which provides a refinement of the general result in \cite[Theorem 1.1]{RV09} in the Gaussian case. As shown in \prettyref{sec:pf.a}, in order for the proof of \prettyref{thm:linear.mix} to work, it is necessary to use ensembles for which $\prob{\sigma_{\min}(\bsB_k) \leq t}$ can be upper bounded asymptotically by $\exp(- k E(t))$, where $E(t) \to \infty$ as $t$ vanishes.
	\label{rmk:RV09}
\end{remark}

The next lemma upper bounds the probability that a Gaussian random matrix shrinks the length of some vector in an affine subspace by a constant factor.
The point of this result is that the bound depends only on the dimension of the subspace but not the basis.
\begin{lemma}
	Let $\bsA$ be a $k \times n$ random matrix with \iid Gaussian entries with zero mean and variance $\frac{1}{n}$. Let $R = \frac{k}{n}$.
	Let $k > m$.
Then for any $m$-dimensional affine subspace $U$ of $\reals^n$,
	\begin{equation}
\prob{\inf_{x \in U \backslash \{0\}}	\frac{\norm{\bsA x}}{\norm{x}} \leq t} \leq F_{k,m}\big(R^{-\frac{1}{2}}t\big).
	\label{eq:affsigma}
\end{equation}
	\label{lmm:affine}
\end{lemma}
\begin{proof}
By definition, 
there exists $v \in \reals^n$ and an $m$-dimensional linear subspace $V$ such that $U = v + V$. 
First assume that $v \notin V$. Then $0 \notin U$. Let $\{\ntok{v_0}{v_m}\}$ be an orthonormal basis for $V'=\Span(v,V)$. Set $\Psi = [\ntok{v_0}{v_m}]$. Then
	\begin{align}
\inf_{x \in U}	\frac{\norm{\bsA x}}{\norm{x}}
= & ~ \min_{x \in V'\backslash\{0\}}	\frac{\norm{\bsA x}}{\norm{x}} \label{eq:infU}\\
= & ~ \min_{y \in \reals^{m+1} \backslash\{0\}}	\frac{\norm{\bsA \Psi y}}{\norm{y}} \\
= & ~ \sigma_{\min}(\bsA \Psi), \label{eq:infU2}
\end{align}
where \prettyref{eq:infU} is due to the following reasoning: since $U \subset V'$, it remains to establish $\inf_{x \in U}	\frac{\norm{\bsA x}}{\norm{x}}
\leq \min_{x \in V'\backslash\{0\}}$. To see this, for any $x \in V'$, we have $x = \alpha v + \beta y$ for some $\alpha,\beta \in \reals$ and $y \in V$. Without loss of generality, we can assume that $\alpha \geq 0$. For each $\tau > 0$, define $x_\tau = (\alpha + \tau) v + \beta y \in V'$ which satisfies $\|x_\tau - x\| \to 0$ as $\tau \to 0$. Then $\frac{x_\tau}{\alpha + \tau} \in U$ and
\begin{equation}
	\frac{\norm{\bsA x}}{\norm{x}} = \lim_{\tau \downarrow 0} \frac{\norm{\bsA \frac{x_{\tau}}{\alpha + \tau}}}{\norm{\frac{x_{\tau}}{\alpha + \tau}}} \geq \inf_{x \in U}	\frac{\norm{\bsA x}}{\norm{x}},
	\label{eq:infU1}
\end{equation}
which, upon minimizing the left-hand side of \prettyref{eq:infU1} over $x \in V'$, implies the desired \prettyref{eq:infU}.
In view of \prettyref{eq:infU2}, \prettyref{eq:affsigma} holds with equality since $\bsA \Psi$ is a $k \times (m+1)$ random matrix with \iid normal entries of zero mean and variance $\frac{1}{n}$.\footnote{Note that the entries in the ensemble in \prettyref{lmm:edelman} have variance inversely proportional to the number of columns.} If $v \in V$, then \prettyref{eq:affsigma} holds with equality and $m+1$ replaced by $m$. The proof is then complete because $m \mapsto F_{k,m}(t)$ is decreasing.
\end{proof}

\begin{lemma}
	Let $T$ be a union of $N$ affine subspaces of $\reals^n$ with dimension not exceeding $m$. Let $\prob{X^n \in T} \geq 1 -\epsilon$. Let $\bsA$ be defined in \prettyref{lmm:affine} independent of $X^n$. Then 
	\begin{equation}
\prob{X^n \in T, \inf_{y \in T \backslash\{X^n\}} \frac{\norm{\bsA (y - X^n)}}{\norm{y-X^n}} \geq t} \geq 1-\epsilon',
	\label{eq:lip1}
\end{equation}
where \begin{equation}
	\epsilon' = \epsilon + N F_{k,m}\big(R^{-\frac{1}{2}}t\big).
	\label{eq:epsilonp}
\end{equation}
Moreover, there exists a subset $E \subset \reals^{k \times n}$ with $\prob{\bsA \in E} \geq 1 - \sqrt{\epsilon'}$, such that for any $\bfK \in E$, there exists a Lipschitz continuous function $g_{\bfK}: \reals^k \to \reals^n$ with $\Lip(g_{\bfK}) \leq \frac{1}{t}$ and
	\begin{equation}
	\prob{g_{\bfA}(\bfA X^n) \neq X^n} \geq 1-\sqrt{\epsilon'}.
	\label{eq:lip11}
\end{equation}
	\label{lmm:lip1}
\end{lemma}
\begin{proof}
By the independence of $X^n$ and $\bsA$,
	\begin{align}
\prob{X^n \in T, \inf_{y \in T \backslash\{X\}} \frac{\norm{\bsA (y - X)}}{\norm{y-X}} \geq t}	
= & ~ \int_{T} P_{X^n}(\diff x)  \prob{\inf_{z \in (T - x) \backslash\{0\}} \frac{\norm{\bsA z}}{\norm{z}} \geq t}	\\
\geq & ~ \prob{X^n \in T} (1 - N  F_{k,m}\big(R^{-\frac{1}{2}}t\big) ) \label{eq:count1}\\
\geq & ~ 1 - \epsilon'.
\end{align}
where \prettyref{eq:count1} follows by applying \prettyref{lmm:affine} to each affine subspace in $T-x$ and the union bound. To prove \prettyref{eq:lip11}, denote by $p(\bfK)$ the probability in the left-hand side of \prettyref{eq:lip1} conditioned on the random matrix $\bsA$ being equal to $\bfK$. By Fubini's theorem and Markov's inequality,
\begin{align}
\prob{p(\bsA) \geq 1 - \sqrt{\epsilon'}} \geq 1 - \sqrt{\epsilon'}.
\end{align}
Put $E = \{\bfK: p(\bfK) \geq 1 - \sqrt{\epsilon'}\}$. For each $\bfK \in E$, define 
\begin{equation}
U_{\bfK} = \sth{x \in T: \inf_{y \in T \backslash\{x\}} \frac{\norm{\bfK (y-x)}}{\norm{y-x}} \geq t} \subset T.
	\label{eq:UK}
\end{equation}
Then, for any $(x, y) \in U_{\bfK}^2$, we have
\begin{equation}
	\norm{\bfK(x-y)} \geq t \norm{x-y},
	\label{eq:invlip}
\end{equation}
which implies that $\restrict{\bfK}{U_{\bfK}}$, the linear mapping $\bfK$ restricted on the set $U_{\bfK}$, is injective. Moreover, its inverse $g_{\bfK}: \bfK(U_{\bfK}) \to U_{\bfK}$ is $\frac{1}{t}$-Lipschitz. By Kirszbraun's theorem \cite[2.10.43]{federer}, $g_{\bfK}$ can be extended to a Lipschitz function on the whole space $\reals^k$ with the same Lipschitz constant. For those $\bfK \notin E$, set $g_{\bfK} \equiv 0$.
Since $\prob{X^n \in U_{\bfK}} \geq 1 - \sqrt{\epsilon'}$ for all $\bfK \in E$, we have
\begin{equation}
\prob{g_{\bfK}(\bfK X^n) \neq X^n} \geq \prob{X^n \in U_{\bfA}, \bfA \in E} \geq  1-\sqrt{\epsilon'},
\end{equation}
completing the proof of the lemma.
\end{proof}

\subsection{Proofs of results in \prettyref{sec:noiseless}}
	\label{sec:pf.a}

	\begin{proof}[Proof of \prettyref{thm:lipconv.f}]
To prove the left inequality in \prettyref{eq:lipconv.f}, denote 
\begin{equation}
C = \{f(x^n) \in \reals^n: g(f(x^n)) = x^n\} \subset \reals^k.
\end{equation}
 Then
\begin{align}
k
\geq & ~ \uBdim(C) \label{eq:m1}\\
\geq & ~ \uBdim(g(C))	\label{eq:m2}\\
\geq & ~ \uBdim^{\epsilon}(P_{X^n}),\label{eq:m3}
\end{align}
where 
\begin{itemize}
	\item \prettyref{eq:m1}: Minkowski dimension never exceeds the ambient dimension;
	\item \prettyref{eq:m2}: Minkowski dimension never increases under Lipschitz mapping \cite[Exercise 7.6, p.108]{mattila};
	\item \prettyref{eq:m3}: by $\prob{X^n \in g(C)} \geq 1- \epsilon$ and \prettyref{eq:ubdim2}.
\end{itemize}

It remains to prove the right inequality in \prettyref{eq:lipconv.f}. By definition of $\uBdim^{\epsilon}$, for any $\delta > 0$, there exists $E$ such that $P_{X^n}(E) \geq 1- \epsilon$ and $\uBdim(E) \geq \uBdim^{\epsilon}(P_{X^n}) - \delta$. 
Since $P_{X^n}$ can be written as a convex combination of $P_{X^n|X^n \in E}$ and $P_{X^n|X^n \notin E}$, applying \cite[Theorem 2]{renyi.ITtrans} yields
\begin{align}
\od(X^n) 
\leq  & ~ \od(P_{X^n|X^n \in E}) P_{X^n}(E) + \od(P_{X^n|X^n \notin E}) (1-P_{X^n}(E))\\
\leq  & ~ \uBdim^{\epsilon}(P_{X^n}) - \delta + \epsilon n, \label{eq:lq}
\end{align}
where \prettyref{eq:lq} holds because the information dimension of any distribution is upper bounded by the Minkowski dimension of its support \cite{falconer2}. By the arbitrariness of $\delta$, the desired result follows.
\end{proof}

	\begin{proof}[Proof of \prettyref{thm:linear.mix}]
Let $P_X$ be a discrete-continuous mixture as in \prettyref{eq:dcmix}. Equation \prettyref{eq:rlinear.dc} is proved in \cite[Theorem 6]{renyi.ITtrans}. The achievability part follows from \prettyref{thm:dimB.finite}, since, with high probability, the input vector is concentrated on a finite union of affine subspaces whose Minkowski dimension is equal to the maximum dimension of those subspaces. The converse part is proved using Steinhaus' theorem \cite{steinhaus}.
	
	It remains to establish the achievability part of \prettyref{eq:rll.dc}: $\rll(X, \epsilon) \leq \gamma$. 
	Fix $R>\gamma$. Fix $\delta, \delta' > 0$ arbitrarily small.
 In view of \prettyref{lmm:lip1}, to prove the achievability of $R$, it suffices to show that, with high probability, $X^n$ lies in the union of exponentially many affine subspaces whose dimensions do not exceed $n R$.

To this end, let $W_i = \indc{X_i \notin \calA}$, where $\calA$ denotes the collection of all atoms of $P_{\rm d}$, which is, by definition, a countable subset of $\reals$. Then $\{W_i\}$ is a sequence of \iid~binary random variables with expectation $\gamma$.
By the weak law of large numbers,
\begin{align}
\frac{1}{n} |\spt{X^n}|
 = \frac{1}{n} \sum_{i=1}^n W_i \, \toprob \gamma. 	\label{eq:WLLN.spt}
\end{align}
where the \emph{generalized support} of $x^n$ is defined as
\begin{equation}
 \spt{x^n} = \{i = \ntok{1}{n}: x_i \notin \calA\}.
\label{eq:spt}
\end{equation}
For each $k \geq 1$, define
\begin{equation}
	\sfT_k = \sth{z^k \in \calA^k: \frac{1}{k} \sum_{i=1}^k \log \frac{1}{P_{\rm d}(z_i)} \leq H(P_{\rm d})+\delta'}.
	\label{eq:Tk}
\end{equation}
Since $H(P_{\rm d}) < \infty$, we have $|\sfT_k| \leq \exp((H(P_{\rm d})+\delta')k)$. Moreover, $P_{\rm d}^k(\sfT_k) \geq 1- \epsilon$ for all sufficiently large $k$, by the weak law of large numbers.

Let $\sft(x^n)$ denote the discrete part of $x^n$, \ie, the vector formed by those $x_i\in \calA$ in increasing order of $i$. Then $\sft(x^n) \in \calA^{n-|\spt{x^n}|}$. Let
\begin{align}
C_n 
= & ~ \left\{ x^n \in \reals^n: \big| |\spt{x^n}| - \gamma n\big| \leq \delta n, \, \sft(x^n) \in \sfT_{n-|\spt{x^n}|} \right\}\\
= & ~ \bigcup_{\substack{S \subset \{\ntok{1}{n}\} \\ ||S| - \gamma n| \leq \delta n}} \bigcup_{z  \in \sfT_{n - |S|}}  \left\{x^n \in \reals^n: \spt{x^n} = S, \, \sft(x^n) = z \right\}. \label{eq:subsets}
\end{align}
Note that each of the subsets in the right-hand side of \prettyref{eq:subsets} is an affine subspace of dimension no more than $(\gamma + \delta)n$. Therefore $C_n$ consists of $N_n$ affine subspaces, with 
\begin{align}
N_n 
\leq & ~ \sum_{k=\floor{(\gamma-\delta) n}}^{\ceil{(\gamma+\delta) n}} \binom{n}{k} |\sfT_{n-k}|\\
\leq & ~ \sum_{k=\floor{(\gamma-\delta) n}}^{\ceil{(\gamma+\delta) n}} \binom{n}{k} \exp((H(P_{\rm d})+\delta')(n-k)),
\end{align}
hence
\begin{equation}
	\limsup_{n\to\infty} \frac{1}{n} \log N_n \leq (H(P_{\rm d}) + \delta') (1-\gamma + \delta) + \max\{h(\gamma+\delta), h(\gamma-\delta)\}.
	\label{eq:logNn}
\end{equation}
Moreover, by \prettyref{eq:WLLN.spt}, for sufficiently large $n$,
\begin{align}
\prob{X^n \in C_n}
= & ~ \sum_{||S| - \gamma n| \leq \delta n} \prob{X^n \in C_n, \spt{X^n} = S} \\
= & ~ \sum_{||S| - \gamma n| \leq \delta n} \prob{\spt{X^n} = S} P_{\rm d}^{n-|S|}(\sfT_{n-|S|}) \\
\geq & ~ \prob{\big| |\spt{X^n}| - \gamma n\big| \leq \delta n} (1-\epsilon) \\
\geq & ~ 1-2\epsilon.
\end{align}
To apply \prettyref{lmm:affine}, it remains to select a sufficiently small but fixed $t$, such that 
\begin{equation}
N_n F_{R n,(\gamma + \delta)n}\big(R^{-\frac{1}{2}}t\big) = o(1)
\end{equation}
as $n\to\infty$. This is always possible, in view of \prettyref{eq:logNn} and \prettyref{lmm:edelman}, by choosing $t > 0$ sufficiently small such that
\begin{equation}
\frac{R (1-\alpha)}{2} \log \frac{R(1-\alpha)^2}{\eexp t^2 } + \frac{R \alpha}{2} \log \alpha
	> (H(P_{\rm d}) + \delta') (1-\gamma + \delta) + \max\{h(\gamma+\delta), h(\gamma-\delta)\},
	\label{eq:lipc}
\end{equation}
where $\alpha = \frac{\gamma + \delta}{R}$.
By the arbitrariness of $\delta$ and $\delta'$, the proof of $\rll(X, \epsilon) \leq \gamma$ is complete. Finally, by \prettyref{thm:lipconv.f}, the Lipschitz constant of the corresponding decoder is upper bounded by $\frac{1}{t}$, which, according to \prettyref{eq:lipc}, can be chosen arbitrary close to the right-hand side of \prettyref{eq:lipconst} by sending both $\delta$ and $\delta'$ to zero,
completing the proof of \prettyref{eq:lipconst}.
\end{proof}



\subsection{Proofs of results in \prettyref{sec:noisy}}
	\label{sec:pf.b}
	\begin{proof}[Proof of \prettyref{thm:Rstar}]
The proof of \prettyref{eq:Rstar} is based on the low-distortion asymptotics of $R_X(D)$ \cite{dembo}:
\begin{gather}
\limsup_{D \downarrow 0} \frac{R_X(D)}{\frac{1}{2} \log \frac{1}{D}} = \od(X), \label{eq:urdim1}
\end{gather}

\emph{Converse}: Fix $R > \Rstar(X)$. By definition, there exits $a > 0$ such that $\Dstar(X, R,\sigma^2) \leq a \sigma^2$ for all $\sigma^2 > 0$. By \prettyref{eq:Dstar.sepa}, 
\begin{equation}
	\frac{R}{\frac{1}{2} \log(1+\sigma^{-2})} \geq R_X(a \sigma^2).
	\label{eq:R12a}
\end{equation}
Dividing both sides by $\frac{1}{2} \log \frac{1}{a \sigma^2}$ and taking $\limsup_{\sigma^2 \to 0}$, we obtain $R > \od(X)$ in view of \prettyref{eq:urdim1}. By the arbitrariness of $R$, we have $\Rstar(X) > \od(X)$.

\emph{Achievability}: Fix $\delta > 0$ arbitrarily and let $R = \od(X) + 2 \delta$. We show that $R \leq \Rstar(X)$, \ie, worst-case noise sensitivity is finite. By \prettyref{rmk:suplim}, this is equivalent to achieving \prettyref{eq:Rstar.equiv}. By \prettyref{eq:urdim1}, there exists $D_0 > 0$ such that for all $D < D_0$,
\begin{equation}
	R_X(D) \leq \frac{\od(X) + \delta}{2} \log \frac{1}{D}.
	\label{eq:R32}
\end{equation}
By \prettyref{thm:distortion.property}, $\Dstar(X, R,\sigma^2) \downarrow 0$ as $\sigma^2 \downarrow 0$. Therefore there exists $\sigma_0^2 > 0$, such that $\Dstar(X, R,\sigma^2) < D_0$ for all $\sigma^2 < \sigma_0^2$. In view of \prettyref{eq:Dstar.sepa} and \prettyref{eq:R32}, we have
\begin{equation}
	\frac{d+2\delta}{2} \log \frac{1}{\sigma^2} = R_X(\Dstar(X, R,\sigma^2)) \leq \frac{\od(X) + \delta}{2} \log \frac{1}{\Dstar(X, R,\sigma^2)},
\end{equation}
\ie,
\begin{equation}
	\Dstar(X, R,\sigma^2) \leq \sigma^{2 \frac{\od(X)+2\delta}{\od(X)+\delta}}
\end{equation}
holds for all $\sigma^2 < \sigma_0^2$. This obviously implies the desired \prettyref{eq:Rstar.equiv}.

We finish the proof by proving \prettyref{eq:Dstar.mix} and \prettyref{eq:xstar.mix}. The low-distortion asymptotic expansion of the rate-distortion function of a discrete-continuous mixture with mean-square error distortion is found in \cite[Corollary 1]{gyorgy}, which refines \prettyref{eq:dembo}:\footnote{In fact $h(\gamma) + (1-\gamma) H(P_{\rm d}) + \gamma h(P_{\rm c})$ is the $\gamma$-dimensional entropy of \prettyref{eq:dcmix} defined by \renyi \cite[Equation (4) and Theorem 3]{renyi}.} as $D \downarrow 0$,
\begin{align}
R_X(D)
= & ~ \frac{\gamma}{2} \log \frac{\gamma}{2 \pi \eexp D} + h(\gamma) + (1-\gamma) H(P_{\rm d}) + \gamma h(P_{\rm c}) + o(1) \\
= & ~ \frac{\gamma}{2} \log \frac{\gamma \, \var(P_{\rm c})}{D} + h(\gamma) + (1-\gamma) H(P_{\rm d}) - \gamma \calD(P_{\rm c}) + o(1) \label{eq:RD.mix}
\end{align}
where $P_X$ is given by \prettyref{eq:dcmix}.
Actually \prettyref{eq:RD.mix} has a natural interpretation: first encode losslessly the \iid Bernoulli sequence $\{\sfA,\sfD,\sfD,\sfD,\sfA,\ldots\}$, where $\sfD$ and $\sfA$ indicate the source realization is in the discrete alphabet or not, respectively. Then use lossless and lossy optimal encoding of $P_{\rm d}$ and $P_{\rm c}$ for the discrete and continuous symbols respectively. What is interesting is that this strategy turns out to be optimal for low distortion.
 Plugging \prettyref{eq:RD.mix} into \prettyref{eq:Dstar.sepa} gives \prettyref{eq:Dstar.mix}, which implies \prettyref{eq:xstar.mix} as a direct consequence.
\end{proof}

\begin{proof}[Proof of \prettyref{thm:noisy.dc}]
Let $R > \gamma$. We show that the worst-case noise sensitivity $\zL(X,R)$ under Gaussian random sensing matrices is finite. We construct a suboptimal estimator based on the Lipschitz decoder in \prettyref{thm:linear.mix}.\footnote{Since we assume that $\var X = 1$, the finite-entropy condition of \prettyref{thm:linear.mix} is satisfied automatically.} Let $\bsA_n$ be a $k \times n$ Gaussian sensing matrix and $g_{\bsA_n}$ the corresponding $L_R$-Lipschitz decoder, such that $k=R n$ and $\prob{E_n} = o(1)$ where $E_n = \{g_{\bsA_n}(\bsA_n X^n) \neq X^n\}$ denotes the error event. Without loss of generality, we assume that $g_{\bsA_n}(0) = 0$. Fix $\tau > 0$. Then
	\begin{align}
& ~ \expect{\big\|g_{\bsA_n}(\bsA_n X^n + \sigma N^k) - X^n\big\|^2} \nonumber \\
\leq & ~ \expect{\big\|g_{\bsA_n}(\bsA_n X^n + \sigma N^k) - X^n\big\|^2 \indc{\comp{E_n}}} \nonumber \\
& ~ + 2 L_R^2 \expect{\big\|\bsA_n X^n + \sigma N^k \big\|^2 \indc{E_n}} + 2 \expect{\big\|X^n\big\|^2\indc{E_n}} \label{eq:mk} \\
\leq & ~ k L_R^2 \sigma^2 + \tau n (2 L_R^2 + 1) \prob{E_n} + 2 \expect{\big\|X^n\big\|^2 \indc{\norm{X^n}^2 > \tau n}} \nonumber \\
& ~ + 2 L_R^2 \expect{\big\|\bsA_n X^n + \sigma N^k \big\|^2 \indc{\norm{\bsA_n X^n + \sigma N^k}^2 > \tau n}}  \label{eq:mk2}
\end{align}
Dividing both sides of \prettyref{eq:mk2} by $n$ and sending $n \to \infty$, we have: for any $\tau > 0$,
\begin{align}
& ~  \limsup_{n \to \infty} \frac{1}{n} \expect{\big\|g_{\bsA_n}(\bsA_n X^n + \sigma N^k) - X^n\big\|^2} \nonumber \\
\leq 	& ~  R L_R^2 \sigma^2 + 2 \sup_n  \frac{1}{n} \expect{\big\|X^n\big\|^2 \indc{\norm{X^n}^2 > \tau n}} \nonumber \\
& ~  + 2 L_R^2 \sup_n \frac{1}{n} \expect{\big\|\bsA_n X^n + \sigma N^k \big\|^2 \indc{\norm{\bsA_n X^n + \sigma N^k}^2 > \tau n}}. 
\label{eq:mk3}
\end{align}
Since $\frac{1}{n} \norm{X^n}^2 \tolp{2} 1$ and $\frac{1}{n} \norm{\bsA_n X^n + \sigma N^k}^2 \tolp{2} R(1+\sigma^2)$, which implies uniform integrability, the last two terms on the right-hand side of \prettyref{eq:mk3} vanish as $\tau \to \infty$. This completes the proof of $\zL(X,R) \leq R L^2_R$.
\end{proof}

\begin{proof}[Proof of \prettyref{thm:RLstar}]
\emph{Achievability}:	We show that $\RL(X) \leq \oscrD(X)$. Fix $\delta > 0$ arbitrarily and let $R = \oscrD(X) + 2 \delta$. Set $s = R \,\sigma^{-2}$ and $\beta = \eta s$. Define
\begin{align}
u(\beta) = & ~ \beta \, \mmse(X, \beta) - R\pth{1 - \frac{\beta}{s} } \label{eq:ubeta} \\
f(\beta) = &~ 	I(X; \sqrt{\beta} X + N) - \frac{R}{2} \log \beta \label{eq:fbeta}\\
g(\beta) = &~ 	f(\beta) + \frac{R \beta}{2 s}, \label{eq:gbeta}
\end{align}
which satisfy the following properties:
\begin{enumerate}
	\item Since $\mmse(X,\cdot)$ is smooth on $(0,\infty)$ \cite[Proposition 7]{mmse.analytic}, $u, f$ and $g$ are all smooth functions on $(0,\infty)$. Additionally, since $\expect{X^2} < \infty$, $u$ is also right-continuous at zero. In particular, by the I-MMSE relationship \cite{guo.immse},
	\begin{equation}
	\dot{f}(\beta) = \frac{\beta \, \mmse(X,\beta) - R}{2\beta}.
	\label{eq:fprime}
\end{equation}
\item For all $0 \leq \beta \leq s$,
\begin{equation}
	f(\beta) \leq g(\beta) \leq f(\beta) + \frac{R}{2}.
	\label{eq:fgbeta}
\end{equation}
\item 
Recalling the scaling law of mutual information in \prettyref{eq:i.renyi}, we have
\begin{equation}
	\limsup_{\beta \to \infty} \frac{f(\beta)}{\log \beta} = \frac{\od(X) - \oscrD(X) - 2 \delta}{2} \leq -\delta,
	\label{eq:fbeta1}
\end{equation}
where the last inequality follows from the sandwich bound between information dimension and MMSE dimension in \prettyref{eq:mmse.renyi}.
\end{enumerate}

Let $\beta_{s}$ be the root of $u$ in $(0, s)$ which minimizes $g(\beta)$. Note that $\beta_{s}$ exists since
$u(0) = -R < 0$, $u(s) = s \, \mmse(X,s) > 0$ and $u$ is continuous on $[0,\infty)$. 
According to \prettyref{eq:om.mmse} in the replica symmetry postulate, 
\begin{equation}
	\DL(X,R,\sigma^2) = \mmse(X, \eta_{s} s),
	\label{eq:DL2}
\end{equation}
where $\eta_{s}$ is the solution of \prettyref{eq:eta} in $(0,1)$ which minimizes \prettyref{eq:eta.min}, denoted by
\begin{equation}
E(\eta) = I(X;\sqrt{\eta s} X+N) + \frac{R}{2}(\eta - 1 - \log \eta). 
	\label{eq:Eta}
\end{equation}
We claim that for any fixed $s$, 
\begin{equation}
\eta_s = \frac{\beta_s}{s}.	
	\label{eq:etab}
\end{equation}
To see this, note that the solutions to \prettyref{eq:eta} are precisely the roots of $u$ scaled by $\frac{1}{s}$. Moreover, since $E(\eta) - g(\beta) = \frac{R}{2}(\log s - 1)$, for any set $A \subset (0,1)$, we have
\begin{equation}
\argmin_{\eta \in A} E(\eta) = \frac{1}{s} \argmin_{\beta \in s A} g(\beta),
\end{equation}
resulting in \prettyref{eq:etab}. Next we focus on the behavior of $\beta_s$ as $s$ grows.

Proving the achievability of $R$ amounts to showing that
\begin{equation}
	\limsup_{\sigma \to 0} \frac{\DL(X,R,\sigma^2)}{\sigma^2} < \infty,
\end{equation}
which, in view of \prettyref{eq:DL2} and \prettyref{eq:etab}, is equivalent to showing that $\beta_s$ grows at least linearly as $s \to \infty$, \ie,
\begin{equation}
	\liminf_{s \to \infty} \frac{\beta_{s}}{s} > 0.
	\label{eq:betao}
\end{equation}
By the definition of $\oscrD(X)$ and \prettyref{eq:fbeta1}, there exists $B > 0$ such that for all $\beta > B$,
\begin{equation}
	\beta \, \mmse(X, \beta) < R - \delta
	\label{eq:pickB1}
\end{equation}
and
\begin{equation}
	f(\beta) \leq -\frac{\delta}{4}\log \beta.
	\label{eq:pickB2}
\end{equation}
In the sequel we focus on sufficiently large $s$. Specifically, we assume that 
\begin{equation}
	s > \frac{R}{\delta} \max\sth{B, \eexp^{-\frac{4}{\delta} \pth{K - \frac{R}{2}}}},
	\label{eq:pickgamma}
\end{equation}
where $K \triangleq \min_{\beta \in [0,B]} g(\beta)$ is finite by the continuity of $g$.


Let 
\begin{equation}
\beta_0 = \frac{\delta s}{R}.
	\label{eq:beta0}
\end{equation}
Then $\beta_0 > B$ by \prettyref{eq:pickgamma}.
By \prettyref{eq:pickB1}, $u(\beta_0) = \beta_0 \, \mmse(X,\beta_0) - R + \delta < 0$. Since $u(s) >0$, by the continuity of $u$ and the intermediate value theorem, there exists $\beta_0 \leq \beta^* \leq s$, such that $u(\beta^*) = 0$. By \prettyref{eq:pickB1}, 
\begin{equation}
\dot{f}(\beta) \leq -\frac{\delta}{2\beta} < 0, \quad  \forall \beta > B. 
	\label{eq:fbeta2}
\end{equation}
Hence $f$ strictly decreases on $(B,\infty)$. Denote the root of $u$ that minimizes $f(\beta)$ by $\beta_{s}'$, which must lie beyond $\beta^*$. Consequently, we have
\begin{equation}
B < \frac{\delta s}{R} = \beta_0 \leq \beta^* \leq \beta_{s}'.
	\label{eq:betagp}
\end{equation}
Next we argue that $\beta_{s}$ cannot differ from $\beta_{s}'$ by a constant factor. In particular, we show that
\begin{equation}
	\beta_{s} \geq \eexp^{-\frac{R}{\delta}} \beta_{s}',
	\label{eq:tworoots}
\end{equation}
which, combined with \prettyref{eq:betagp}, implies that
\begin{equation}
	\frac{\beta_{s}}{s} \geq \frac{\delta}{R} \eexp^{-\frac{R}{\delta}}
\end{equation}
for all $s$ that satisfy \prettyref{eq:pickgamma}. This yields the desired \prettyref{eq:betao}. We now complete the proof by showing \prettyref{eq:tworoots}. First, we show that that $\beta_{s} > B$. This is because
\begin{align}
g(\beta_{s})
\leq & ~ 	g(\beta_{s}') \label{eq:pck1}\\
= & ~ f(\beta_{s}') + \frac{R \beta_{s}'}{2 s} \label{eq:pck2}	\\
\leq & ~ f(\beta_0) + \frac{R }{2} 	\label{eq:pck3}\\
\leq & ~ 	-\frac{\delta}{4} \log \frac{\delta s}{R} + \frac{R}{2} \label{eq:pck4}\\
< & ~ K \\
= & ~ \min_{\beta \in [0,B]} g(\beta). \label{eq:pck5}
\end{align}
where 
\begin{itemize}
	\item \prettyref{eq:pck1}: by definition, $\beta_{s}$ and $\beta_{s}'$ are both roots of $u$ and $\beta_{s}$ minimizes $g$ among all roots;
	\item \prettyref{eq:pck3}: by \prettyref{eq:betagp} and the fact that $f$ is strictly decreasing on $(B,\infty)$;
	\item \prettyref{eq:pck4}: by \prettyref{eq:pickB2} and \prettyref{eq:beta0};
	\item \prettyref{eq:pck5}: by \prettyref{eq:pickgamma}.
\end{itemize}
Now we prove \prettyref{eq:tworoots} by contradiction. Suppose $\beta_{s} < \eexp^{-\frac{R}{\delta}}\beta_{s}'$. Then
\begin{align}
g(\beta_{s}') - g(\beta_{s})
= & ~ \frac{R}{2 s} (\beta_{s}' - \beta_{s}) +  f(\beta_{s}') - f(\beta_{s}) \label{eq:pk1}\\
\leq & ~ \frac{R}{2} + \int_{\beta_{s}}^{\beta_{s}'} \dot{f}(\tau) \diff \tau \label{eq:pk2}	\\
\leq & ~ \frac{R}{2} - \frac{\delta}{2} \log \frac{\beta_{s}'}{\beta_{s}} \label{eq:pk3}	\\
< & ~ 0,
\end{align}
contradicting \prettyref{eq:pck1}, where \prettyref{eq:pk3} is due to \prettyref{eq:fbeta2}.

\emph{Converse:} We show that $\RL(X) \geq \uscrD(X)$. Recall that $\RL(X)$ is the minimum rate that guarantees that the reconstruction error $\DL(X,R,\sigma^2)$ vanishes according to $O(\sigma^2)$ as $\sigma^2 \to 0$. In fact, we will show a stronger result: as long as $\DL(X,R,\sigma^2) = o(1)$ as $\sigma^2 \to 0$, we have $R \geq \uscrD(X)$. By \prettyref{eq:DL2}, $\DL(X,R,\sigma^2) = o(1)$ if and only if $\beta_{s} \to \infty$. Since $u(\beta_{s}) = 0$, we have
\begin{align}
R
\geq & ~ \limsup_{s \to \infty} R\pth{1 - \frac{\beta_{s}}{s} } \\
= & ~ \limsup_{s \to \infty} \beta_{s} \, \mmse(X, \beta_{s}) \label{eq:pk5}	\\
\geq & ~ \liminf_{\beta \to \infty} \beta \, \mmse(X, \beta) 	\\
= & ~ \uscrD(X).
\end{align}

\emph{Asymptotic noise sensitivity:}
Finally, we prove \prettyref{eq:sens.mix}. Assume that $\scrD(X)$ exists, \ie, $\scrD(X) = d(X)$, in view of \prettyref{eq:mmse.renyi}. By definition of $\scrD(X)$, we have
\begin{equation}
	\mmse(X, \beta) = \frac{\scrD(X)}{\beta} + \smallo{\frac{1}{\beta}}, \quad \beta \to \infty.
	\label{eq:mmse.beta}
\end{equation}
As we saw in the achievability proof, whenever $R > \scrD(X)$, \prettyref{eq:betao} holds, \ie, $\eta_{s} = \Omega(1)$ as $s \to \infty$. Therefore, as $s \to \infty$, we have
\begin{equation}
	\frac{1}{\eta_{s}} = 1 + \frac{s}{R} \mmse(X,\eta_{s} s) = 1 + \frac{\scrD(X)}{\eta_{s} R}(1 + o(1)),
\end{equation}
\ie,
\begin{equation}
	\eta_{s} = 1 - \frac{\scrD(X)}{R} + o(1).
\end{equation}
By the replica symmetry postulate \prettyref{eq:om.mmse},
\begin{align}
\DL(X,R,\sigma^2)
= & ~ \mmse(X,\eta_{s} s ) \\
= & ~ \frac{1-\eta_{s}}{\eta_{s}} \sigma^2\\
= & ~ \frac{\scrD(X)}{R - \scrD(X)} \sigma^2 (1	+ o(1)).
\end{align}

\end{proof}

\begin{remark}
Note that $\beta_{s}$ is a \emph{subsequence} parametrized by $s$, which may take only a restricted subset of values. In fact, even if we impose the requirement that $\DL(X,R,\sigma^2) = O(\sigma^2)$, it is still possible that the limit in \prettyref{eq:pk5} lies strictly between $\uscrD(X)$  and $\oscrD(X)$. For example, if $X$ is Cantor distributed as defined in \prettyref{eq:cantor}, then it can be shown that the limit in \prettyref{eq:pk5} approaches the information dimension $d(X) = \log_3 2$.
\end{remark}

\begin{remark}[Multiple solutions in the replica symmetry postulate]
Solutions to \prettyref{eq:eta} in the replica symmetry postulate and 
to the following equation in $\beta$
\begin{equation}
	\beta \, \mmse(X,\beta) = R - \sigma^2 \beta.
	\label{eq:rss}
\end{equation}
differ only by a scale factor of $\frac{\sigma^2}{R}$.
Next we give an explicit example where \prettyref{eq:rss} can have \emph{arbitrarily many} solutions. Let $X$ be Cantor distributed as defined in \prettyref{eq:cantor}. According to \cite[Theorem 16]{mmse.dim.IT}, $\beta \mapsto \beta \, \mmse(X,\beta)$ oscillates in $\log_3 \beta$ with period two, as shown in \prettyref{fig:cantor} in a linear-log plot. Therefore, as $\sigma^2 \to 0$, the number of solutions to \prettyref{eq:rss} grows unbounded according to $\Theta\left(\log \frac{1}{\sigma^2}\right)$. In fact, in order for \prettyref{thm:RLstar} to hold, it is crucial that the replica solution be given by the solution that \emph{minimizes} the free energy \prettyref{eq:eta.min}. 
	\label{rmk:cantor}
\end{remark}

\begin{figure}[ht] 
	\centering
	\includegraphics[width=.6\columnwidth]{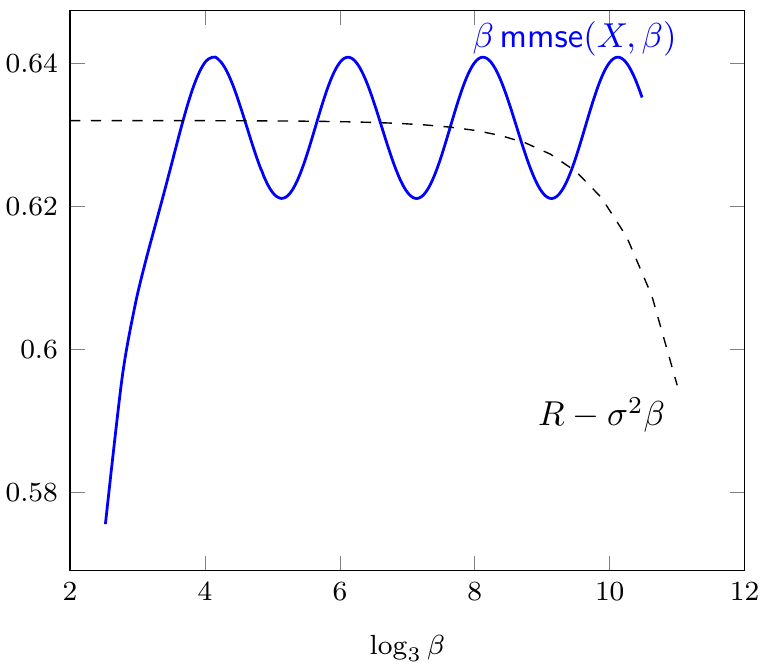}
	\caption{Multiple solutions to \prettyref{eq:rss} in the replica symmetry postulate, with Cantor distributed $X$, $R = 0.632$ and $\sigma^2=3^{-14}$.}
	\label{fig:cantor}
\end{figure}

\section{Concluding remarks}
\label{sec:conc}

In the compressed sensing literature it is common to guarantee that for any individual sparse input the matrix will likely lead to reconstruction, or, alternatively, that a single matrix will work for all possible signals. 
As opposed to this worst-case (Hamming) approach, in this paper we adopt a statistical (Shannon) framework for compressed sensing by modeling input signals as random processes rather than individual sequences.
As customary in information theory, it is advisable to initiate the study of fundamental limits assuming independent identically distributed information sources.
Naturally, this entails substantial loss of practical relevance, so generalization to sources with memory is left for future work. 
The fundamental limits apply to the asymptotic regime of large signal dimension, although a number of the results in the noiseless case are in fact non-asymptotic (see, \eg, Theorems \ref{thm:dimB.finite} and \ref{thm:lipconv.f}).

We have investigated the phase transition thresholds (minimum measurement rate) of reconstruction error probability (noiseless observations) and normalized MMSE (noisy observations) achievable by optimal nonlinear, optimal linear, and random linear encoders combined with the corresponding optimal decoders (\ie conditional mean estimates).
For discrete-continuous mixtures, which are the most relevant for compressed sensing applications, the optimal phase transition threshold is shown to be the information dimension of the input, \ie, the weight of the analog part, regardless of the specific discrete and absolutely continuous component. 
The universal optimality of random sensing matrices with non-Gaussian i.i.d.\ entries in terms of phase transition thresholds is still unknown. 
The phase-transition thresholds of popular decoding algorithms (\eg, LASSO or AMP decoders) turn out to be far from the optimal boundary. 
In a recent preprint \cite{DJM11}, it is shown that using random sensing matrices constructed from spatially coupled error-correcting codes \cite{KP10} and the corresponding AMP decoder, the information dimension can be achieved
under mild conditions, which are optimal in view of the results in \cite{renyi.ITtrans}. Designing deterministic sensing matrices that attain the optimal thresholds remains an outstanding challenge. 

In contrast to the Shannon theoretic limits of lossless and lossy compression of discrete sources,
one of the lessons drawn from the results in this paper and \cite{renyi.ITtrans}
is that compressed sensing of every (memoryless) process taking values on finite or countably infinite
alphabets can be accomplished at zero rate, as long as the observations are noiseless.
In fact, we have even shown in \prettyref{thm:dimB.finite} a non-asymptotic embodiment of this conclusion based on
a probabilistic extension of the embeddability of fractal sets.
In the case of noisy observations, the same insensitivity to the actual discrete signal distribution holds
as far as the phase transition threshold is concerned. However, in the non-asymptotic regime
(\ie  for given signal dimension and signal-to-noise-ratio) the optimum rate-distortion tradeoff
will indeed depend on the signal distribution.


In this paper we have assumed a Bayesian setup where the input is \iid with common distribution known to both the encoder and the decoder. In contrast, the minimax formulation in \cite{maleki.noise.sens,DJoM11,DJMM11} assumes that the input distribution is a discrete-continuous mixture whose  discrete component is known to be a point mass at zero, while the continuous component, \ie, the prior of the non-zero part, is unknown. Minimax analyses were carried out for LASSO and AMP algorithms \cite{maleki.noise.sens}, where the minimum and maximum are with respect to the parameter of the algorithm and the non-zero prior, respectively. The results in \prettyref{sec:compare} demonstrate that the LASSO and AMP algorithms do not attain the fundamental limit achieved by the optimal decoder in the Bayesian setup. However, it is possible to improve performance if the input distribution is known to the reconstruction algorithm. For example, the message passing decoder in \cite{DJM11} that achieves the optimal phase transition threshold is a variant of the AMP algorithm where the denoiser is replaced by the Bayesian estimator (conditional mean) of the input under additive Gaussian noise. See also \cite[Section 6.2]{DMM11how} about how to incorporate the prior information into the AMP algorithm. 


One of our main findings is \prettyref{thm:noisy.dc} which shows that \emph{\iid Gaussian} sensing matrices achieve the same phase-transition threshold as optimal nonlinear encoding, for any discrete-continuous mixture. This result is universal in the sense that it holds for arbitrary noise distributions with finite non-Gaussianness. Moreover, the fundamental limit depends on the input statistics only through the weight of the analog component, regardless of the specific discrete and continuous components. The argument used in the proof of \prettyref{thm:noisy.dc} relies crucially on the Gaussianness of the sensing matrices because of two reasons:
\begin{itemize}
	\item The upper bound on the distribution function of the least singular value in \prettyref{lmm:edelman} is a direct consequence of the upper bound on its density (due to Edelman \cite{Edelman88}), which is only known in the Gaussian case. In fact, we only need that the exponent in \prettyref{eq:sigma1} diverges as $t \to 0$. It is possible to generalize this result to other sub-Gaussian ensembles with densities by adapting the arguments in \cite[Theorem 1.1]{RV09}. However, it should be noted that in general \prettyref{lmm:edelman} does not hold for discrete ensembles (\eg Rademacher), because the least singular value always has a mass at zero with a fixed exponent;
	\item Due to the rotational invariance of the Gaussian ensemble, the result in \prettyref{lmm:affine} does not depend on the basis of the subspace.
\end{itemize}

Another contribution of this work is the rigorous proof of the phase transition thresholds for mixture distributions. Furthermore, based on the MMSE dimension results in \cite{mmse.dim.IT}, we have shown in \prettyref{sec:replica} that these conclusions coincide with previous predictions put forth on the basis of replica-symmetry heuristics.

One interesting direction is to investigate the optimal sensing matrix in a minimax sense. While our \prettyref{thm:noisy.dc} shows that optimized sensing matrices (or even non-linear encoders) do not improve the phase transition threshold for Gaussian sensing matrices, it should be interesting to ascertain whether this conclusion carries over to the minimax setup, \ie, whether it is possible to lower the minimax phase transition threshold of the noise sensitivity achieved by Gaussian sensing matrices and LASSO or AMP reconstruction algorithms computed in \cite{maleki.noise.sens} by optimizing the sensing matrix subject to the Frobenius-norm constraint in \prettyref{eq:power.tr}.


\renewcommand{\appendixname}{Appendix}
\appendices

\section{Proof of the middle inequality in \prettyref{eq:rank}}
	\label{app:middle}
		We show that for any fixed $\epsilon > 0$,
	\begin{equation}
\DLstar(X, R,\sigma^2) \leq \DL(X, R, (1+\epsilon)^2\sigma^2).	
	\label{eq:dlls}
\end{equation}
By the continuity of $\sigma^{-2} \mapsto \DLstar(X, R,\sigma^2)$ proved in \prettyref{thm:distortion.property}, $\sigma^{2} \mapsto \DLstar(X, R,\sigma^2)$ is also continuous. Therefore sending $\epsilon \downarrow 0$ in \prettyref{eq:dlls} yields the second inequality in \prettyref{eq:rank}. To show \prettyref{eq:dlls}, recall that $\bsA$ consists of \iid entries with zero mean and variance $\frac{1}{n}$. Since $k = n R$, $\frac{1}{k} \Fnorm{\bsA}^2 \toprob 1$ as $n \to \infty$, by the weak law of large numbers.  Therefore $\prob{\bsA \in E_n} \to 1$ where
\begin{equation}
	E_n \triangleq \sth{\bfA: \Fnorm{\bfA}^2 \leq k (1+\epsilon)^2}.
	\label{eq:Ens}
\end{equation}
Therefore
\begin{align}
	& ~ \DL(X,R,(1+\epsilon)\sigma^2)	\nonumber \\
	= & ~ \limsup_{n \to \infty} \frac{1}{n} \mmse\pth{X^n |\bsA  X^n + (1+\epsilon)^2\sigma^2 N^k, \bsA} 	\\
= & ~ \limsup_{n \to \infty} \frac{1}{n} \mmse\pth{X^n\Big| \frac{\bsA}{1+\epsilon}  X^n +  N^k, \bsA}	\\
\geq & ~ \limsup_{n \to \infty} \frac{\prob{\bsA \in E_n}}{n} \expect{\mmse\pth{X^n\Big| \frac{\bsA}{1+\epsilon}  X^n +  N^k, \bsA} \Big | \bsA \in E_n}	\label{eq:Ens1}\\
= & ~\DLstar(X, R,\sigma^2),
\end{align}
where \prettyref{eq:Ens1} holds because $\frac{\bfA}{1+\epsilon}$ satisfies the power constraint for any $\bfA \in E_n$.

	\section{Distortion-rate tradeoff of Gaussian inputs}
	\label{app:gaussian.linear}
	In this appendix we show the expressions \prettyref{eq:Dstar.g} -- \prettyref{eq:DL.g} for the minimal distortion, thereby completing the proof of \prettyref{thm:worstG}
	
		\subsection{Optimal encoder}
		Plugging the rate-distortion function of the standard Gaussian \iid random process with mean-square error distortion
		\begin{equation}
	R_{\XGn}(D) = \frac{1}{2} \log^+ \frac{1}{D}
	\label{eq:RD.g}
\end{equation}
		 into \prettyref{eq:Dstar.sepa} yields the equality in \prettyref{eq:Dstar.g}.

	\subsection{Optimal linear encoder}
The minimal distortion $\DLstar(X, R,\sigma^2)$ achievable with the optimal linear encoder can be obtained using the finite-dimensional results in \cite[Equations (31) -- (35)]{LP76}, which are obtained for Gaussian input and noise of arbitrary covariance matrices. We include a proof for the sake of completeness. 	

Denote the sensing matrix by $\bfH$. Since $X^n$ and $Y^k = \bfH X^n + \sigma N^k$ are jointly Gaussian, the conditional distribution of $X^n$ given $Y^k$ is $\calN(\hat{X}^n, \boldsymbol{\Sigma}_{X^n|Y^k})$, where
\begin{align}
\hat{X}^n = & ~ \bfH\trans (\bfH \bfH \trans + \sigma^2 \bfI_k)^{-1} Y^k \label{eq:condmean}\\
\boldsymbol{\Sigma}_{X^n|Y^k} = & ~ \bfI_n - \bfH\trans (\bfH \bfH \trans + \sigma^2 \bfI_k)^{-1} \bfH \\
= & ~ (\bfI_n + \sigma^{-2}\bfH\trans \bfH )^{-1}.
\end{align}
where we used the matrix inversion lemma. Therefore, the optimal estimator is linear, given by \prettyref{eq:condmean}. Moreover,
\begin{align}
\mmse(X^n | Y^k)
= & ~ \Tr(\boldsymbol{\Sigma}_{X^n|Y^k})\\
= & ~ \Tr( (\bfI_n + \sigma^{-2} \bfH\trans \bfH)^{-1} ) \label{eq:mmseH}.
\end{align}

Choosing the best encoding matrix $\bfH \in \reals^{k \times n}$ boils down to the following optimization problem:
\begin{equation}
	\begin{aligned}
\min & ~ \Tr( (\bfI_n + \sigma^{-2} \bfH\trans \bfH )^{-1} )	\\
{\rm s.t.} & ~ \Tr( \bfH\trans \bfH ) \leq k 
	\end{aligned}
\end{equation}
Let $\bfH\trans \bfH = \bfU\trans \Lambda \bfU$, where $\bfU$ is an $n \times n$ orthogonal matrix and $\Lambda$ is a diagonal matrix consisting of the eigenvalues of $\bfH\trans \bfH$, denoted by $\{\ntok{\lambda_1}{\lambda_n}\} \subset \reals_+$. Then
\begin{align}
\Tr( (\bfI_n + \sigma^{-2} \bfH\trans \bfH )^{-1} )
= & \sum_{i=1}^n \frac{1}{1 + \sigma^{-2} \lambda_{i}} \label{eq:tr.1}\\
\geq & \frac{n}{1 + \sigma^{-2} \frac{\Tr(\bfH\trans \bfH)}{n}} \label{eq:tr.2} \\
\geq & \frac{n}{1 +  R \sigma^{-2}} \label{eq:tr.4} 
\end{align}
where \prettyref{eq:tr.2} follows from the strict convexity of $x \mapsto \frac{1}{1 + \sigma^{-2} x}$ on $\reals_+$ and $\Tr(\bfH\trans \bfH) = \sum_{i=1}^n \lambda_i$, while \prettyref{eq:tr.4} is due to the power constraint and $R = \frac{k}{n}$.
Hence
\begin{equation}
	\DLstar(\XGn,R,\sigma^2) \geq \frac{1}{1 +  R \sigma^{-2}}.
	\label{eq:optimal.mmse.linear.1}
\end{equation}
Next we consider two cases separately:
\begin{enumerate}
	\item $R \geq 1 (k \geq n)$: the lower bound in \prettyref{eq:optimal.mmse.linear.1} can be achieved by
	\begin{equation}
	\bfH = \left[\begin{matrix} \sqrt{R} \bfI_n \\ 0\end{matrix}\right].
	\label{eq:opt.H.1}
\end{equation}
\item $R < 1 (k < n)$: the lower bound in \prettyref{eq:optimal.mmse.linear.1} is \emph{not} achievable. This is because to achieve equality in \prettyref{eq:tr.2}, all $\lambda_i$ must be equal to $R$; however, $\rank(\bfH\trans \bfH) \leq \rank(\bfH) \leq k < n$ implies that at least $n-k$ of them are zero. Therefore the lower bound \prettyref{eq:tr.4} can be further improved to:
\begin{align}
\Tr( (\bfI_n + \sigma^{-2} \bfH\trans \bfH )^{-1} )
= 		& ~ n-k + \sum_{\lambda_i > 0} \frac{1}{1 + \sigma^{-2} \lambda_i} \\
\geq 	& ~ n-k + \frac{k}{1 + \sigma^{-2} \frac{\Tr(\bfH\trans \bfH)}{k}} \\
\geq 	& ~ n - \frac{k }{1 +  \sigma^2}.
\end{align}
Hence when $R < 1$,
\begin{equation}
	\DLstar(\XGn,R,\sigma^2) \geq 1 - \frac{R}{1 +  \sigma^2},
	\label{eq:optimal.mmse.linear.2}
\end{equation}
which can be achieved by
\begin{equation}
	\bfH = \left[\begin{matrix} \bfI_k &  0 \end{matrix}\right],
	\label{eq:opt.H.2}
\end{equation}
that is, simply keeping the first $k$ coordinates of $X^n$ and discarding the rest.
\end{enumerate}
Therefore the equality in \prettyref{eq:DLstar.g} is proved.

\subsection{Random linear encoder}

We compute the distortion $\DL(X, R,\sigma^2)$ achievable with random linear encoder $\bsA$. 
Recall that $\bsA$ has \iid entries with zero mean and variance $\frac{1}{n}$. By \prettyref{eq:mmseH},
\begin{align}
\frac{1}{n} \mmse(X^n | \bsA X^n + \sigma N^k, \bsA)
= & ~ \frac{1}{n} \expect{\Tr( (\bfI_n + \sigma^{-2} \bsA\trans \bsA)^{-1} )} \\
= & ~ \frac{1}{n} \expect{\sum_{i=1}^n \frac{1}{1 + \sigma^{-2} \lambda_i}}, \label{eq:mp1}
\end{align}
where $\{\ntok{\lambda_1}{\lambda_n}\}$ are the eigenvalues of $\bsA \trans \bsA$.

As $n \to \infty$, the empirical distribution of 
the eigenvalues of $\frac{1}{R} \bsA \trans \bsA$
converges weakly to the Mar\u cenko-Pastur law almost surely
\cite[Theorem 2.35]{TV04}:
\begin{equation}
	\nu_R(\diff x)  = (1-R)^+ \delta_0(\diff x) + \frac{\sqrt{(x - a)(x - b)}}{2 \pi c x} \mathbf{1}_{[a,b]}(x) \diff x
	\label{eq:f.lambda}
\end{equation}
where
\begin{equation}
	c = \frac{1}{R}, a = (1 - \sqrt{c})^2,  b = (1 + \sqrt{c})^2.
\end{equation}
Since $\lambda \mapsto \frac{1}{1 + \sigma^{-2} \lambda}$ is continuous and bounded, applying the dominated convergence theorem to \prettyref{eq:mp1} and integrating with respect to $\nu_R$ gives
\begin{align}
\DL(\XGn, R,\sigma^2)
=	& ~ \lim_{n \to \infty} \frac{1}{n} \mmse(X^n | \bsA X^n + \sigma N^k, \bsA) 	\nonumber \\
= & ~ \int \frac{1}{1 + \sigma^{-2} R x} \nu_R(\diff x)\\
=	& ~ \frac{1}{2} \left(1-R-\sigma ^2+\sqrt{(1-R)^2+2 (1+R) \sigma ^2+\sigma ^4}\right), \label{eq:DLG}
\end{align}
%
%
where \prettyref{eq:DLG} follows from \cite[(1.16)]{TV04}. 

Next we verify that the formula \prettyref{eq:om.mmse} in the replica symmetry postulate which was based on replica calculations coincides with \prettyref{eq:DLG} in the Gaussian case.
Since in this case $\mmse(\XGn, \snr) = \frac{1}{1+\snr}$, \prettyref{eq:eta} becomes
\begin{align}
\frac{1}{\eta}
= & ~ 1 + \frac{1}{\sigma^2} \mmse(X, \eta R \sigma^{-2}) \label{eq:eta.G.1} \\
= & ~ 1 + \frac{1}{\sigma^2 + \eta R} \label{eq:eta.G.2}
\end{align}
whose unique positive solution is given by
\begin{equation}
\eta_{\sigma} = \frac{R-1-\sigma^2+\sqrt{(1-R)^2+2 (1+R) \sigma ^2+\sigma ^4}}{2 R}
	\label{eq:eta.gamma.G}
\end{equation}
which lies in $(0,1)$. According to \prettyref{eq:om.mmse},
\begin{align}
\DL(\XGn, R,\sigma^2)
= & ~ \mmse(\XGn, \sigma^{-2} \eta_{\sigma}) \\
= & ~ \frac{1}{1+\sigma^{-2} \eta_{\sigma}} \\
= & ~ \frac{2 \sigma ^2}{R-1+\sigma ^2+\sqrt{(1-R)^2+2 (1+R) \sigma ^2+\sigma ^4}},
\end{align}
which can be verified, after straightforward algebra, to coincide with \prettyref{eq:DLG}.

\section{LASSO noise sensitivity for fixed input distributions}
	\label{app:lasso}
Based on the results in \cite{bayati.lasso}, in this appendix we show that the asymptotic noise sensitivity of LASSO is given by \prettyref{eq:sens.lasso}. 
Let $R=\frac{k}{n}$, and let $\bsA$ denote a $k\times n$ random matrix with \iid entries distributed according to $\calN(0,\frac{1}{n})$. Then $R^{-\frac{1}{2}} \bsA$ has $\calN(0,\frac{1}{k})$ entries, to which the result in \cite{bayati.lasso} applies. Let $\tilde{g}(y,\bsA;\lambda)$ denote the LASSO procedure with penalization parameter $\lambda$ defined in \prettyref{eq:lasso}, which satisfies the following scaling-invariant property
\begin{equation}
	\tilde{g}(t y, t \bsA; t\lambda) = \tilde{g}(y, \bsA;\lambda)
	\label{eq:lasso.scaling}
\end{equation}
for any $t > 0$.
By \cite[Corollary 1.6]{bayati.lasso} and \prettyref{eq:lasso.scaling}, the MSE achieved by the LASSO decoder is given by
\begin{align}
	& ~ D^{(\lambda)}(X,R,\sigma^2) 	\nonumber \\
= & ~ \lim_{n \to \infty} \frac{1}{n} \expect{\big\|X^n - \tilde{g}(\bsA X^n + \sigma  N^k; \lambda)\big\|^2} 	\\
= & ~ \lim_{n \to \infty} \frac{1}{n} \expect{\big\|X^n - \tilde{g}(R^{-\frac{1}{2}} \bsA X^n + R^{-\frac{1}{2}} \sigma  N^k; \lambda R^{-\frac{1}{2}})\big\|^2} 	\\
= & ~ R \tau_*^2 - \sigma^2, \label{eq:lassomse}
\end{align}
with $\tau_*^2$ being the unique solution to the following equation in $\tau^2$:
\begin{equation}
	R \tau^2 = \sigma^2 + \expect{(\eta(X+\tau N; \alpha \tau) - X)^2},
	\label{eq:bm}
\end{equation}
where $\eta(\cdot; \cdot): \reals \times \reals_+ \to \reals$ is the soft-thresholding estimator
\begin{equation}
	\eta(x;\theta)=(x-\theta)\indc{x \geq \theta} + (x+\theta)\indc{x \leq -\theta}
	\label{eq:soft}
\end{equation}
and $\alpha = \alpha(\lambda R^{-\frac{1}{2}})$ with $\alpha(\cdot)$ being the strictly increasing function defined in \cite[p. 1999]{bayati.lasso}. Therefore optimizing $D^{(\lambda)}(X,R,\sigma^2)$ over $\lambda$ is equivalent to optimizing over $\alpha$.

Next we assume that $X$ is distributed according to the mixture \prettyref{eq:PXQ}, where $Q$ is an arbitrary probability measure such that $Q(\{0\})=0$. We analyze the weak-noise behavior of $D^{(\lambda)}(X,R,\sigma^2)$ when $R > \sfR_{\scriptscriptstyle \pm}(\gamma)$ defined in \prettyref{eq:rp}.
We show that for fixed $\alpha > 0$,
\begin{align}
& ~ \expect{(\eta(X+\tau N; \alpha \tau) - X)^2} \nonumber \\
= & ~ (\gamma(1+\alpha^2) + 2(1-\gamma)((1+\alpha^2) \Phi(-\alpha)-\alpha \varphi(\alpha)))
\tau^2(1+o(1))
	\label{eq:wtau}
\end{align}
as $\tau \to 0$. Assembling \prettyref{eq:rp}, \prettyref{eq:lassomse}, \prettyref{eq:bm} and \prettyref{eq:wtau}, we obtain the formula for the asymptotic noise sensitivity of optimized LASSO:
	\begin{equation}
\tilde{\xi}(X,R) = \inf_{\lambda} \lim_{\sigma^2 \to 0}\frac{D^{(\lambda)}(X,R,\sigma^2)}{\sigma^2} = \frac{\sfR_{\scriptscriptstyle \pm}(\gamma)}{R-\sfR_{\scriptscriptstyle \pm}(\gamma)},
\label{eq:xilasso}
\end{equation}
which holds for any $Q$ with no mass at zero.

We now complete the proof of \prettyref{eq:xilasso} by establishing \prettyref{eq:wtau}. Let $X' \sim Q$. By \prettyref{eq:PXQ},
\begin{align}
\expect{(\eta(X+\tau N; \alpha \tau) - X)^2} 
= & ~ (1-\gamma) \expect{\eta^2(\tau N; \alpha \tau)} \\
& ~ + \gamma \, \expect{(\eta(X'+\tau N; \alpha \tau) - X')^2},
	\label{eq:wtau1}
\end{align}
where
\begin{align}
\expect{\eta^2(\tau N; \alpha \tau)}
= & ~ 2 \tau^2 \, \expect{(N-\alpha)^2 \indc{N \geq \alpha}} \\
= & ~ 2((1+\alpha^2) \Phi(-\alpha)-\alpha \varphi(\alpha))
	\label{eq:term1}
\end{align}
and
\begin{align}
	& ~ \expect{(\eta(X'+\tau N; \alpha \tau) - X')^2}	\nonumber \\
= & ~ \tau^2 (\expect{(N-\alpha)^2 \indc{X'+\tau N \geq \alpha \tau}} + \expect{(N+\alpha)^2 \indc{X'+\tau N \leq -\alpha \tau}}) \label{eq:term21}\\
	& ~ + \expect{X'^2 \indc{|X'+\tau Z| \leq \alpha \tau}} \label{eq:term22}.
\end{align}
Since $\indc{X'+\tau N \geq \alpha \tau} \toas \indc{X'\geq 0}$, $\indc{X'+\tau N \leq -\alpha \tau} \toas \indc{X'\leq 0}$ and $\prob{X'=0}=0$, applying the bounded convergence theorem to the right-hand side of \prettyref{eq:term21} yields $\tau^2(1+\alpha^2)(1+o(1))$. It remains to show that the term in \prettyref{eq:term22} is $o(\tau^2)$. Indeed, as $\tau\to 0$,
\begin{align}
	& ~ \tau^{-2} \expect{X'^2 \indc{|X'+\tau Z| \leq \alpha \tau}}	\nonumber \\
= & ~ \tau^{-2} \expect{X'^2 \pth{\Phi\pth{-\frac{X'}{\tau}+\alpha} - \Phi\pth{-\frac{X'}{\tau}+\alpha}}}	\\
\leq & ~ 2\alpha \tau^{-2} \expect{X'^2 \varphi\pth{-\frac{|X'|}{\tau}+\alpha}}	\\
= & ~ o(1),
\end{align}
where we have applied the bounded convergence theorem since 
\begin{equation}
\tau^{-2} X'^2 \varphi\pth{-\frac{|X'|}{\tau}+\alpha} \leq \max_{t\geq 0} t^2\varphi(-t+\alpha) = \frac{(\alpha+\sqrt{8+\alpha^2})^2}{4} \varphi\pth{\frac{\alpha-\sqrt{8+\alpha^2}}{2}}
\end{equation}
and $\tau^{-2} X'^2 \varphi\pth{-\frac{|X'|}{\tau}+\alpha} \toas 0$ as $\tau\to0$ because $\prob{|X'|>0}=1$, completing the proof of \prettyref{eq:xilasso} if $R > \sfR_{\scriptscriptstyle \pm}(\gamma)$. In the case $R < \sfR_{\scriptscriptstyle \pm}(\gamma)$, the same reasoning yields that $\liminf_{\sigma^2\to 0}D^{(\lambda)}(X,R,\sigma^2) > 0$ for any choice of $\lambda$.

\section*{Acknowledgment}
\addcontentsline{toc}{section}{Acknowledgment}
The paper has benefited from thorough suggestions by the anonymous reviewers. The authors also thank Arian Maleki for stimulating discussions especially on the LASSO and AMP algorithms.


\end{document}